\documentclass[11pt]{article}

\usepackage[margin=1in]{geometry}
\usepackage{amsmath,amssymb,amsthm,mathtools}
\usepackage{booktabs}
\usepackage{graphicx}
\usepackage{pgfplots}
\pgfplotsset{compat=1.18}
\usepackage{enumitem}
\usepackage{float}
\usepackage[ruled,linesnumbered,vlined]{algorithm2e}
\usepackage{xcolor}
\usepackage{microtype}
\usepackage[numbers,sort&compress]{natbib}
\usepackage[colorlinks=true,linkcolor=blue!55!black,citecolor=blue!55!black,urlcolor=blue!60!black]{hyperref}
\usepackage{tikz}
\usetikzlibrary{arrows.meta,positioning,calc}

\newtheorem{theorem}{Theorem}[section]
\newtheorem{lemma}[theorem]{Lemma}

\newtheorem{corollary}[theorem]{Corollary}
\theoremstyle{definition}
\newtheorem{definition}[theorem]{Definition}

\newcommand{\MMS}{\operatorname{MMS}}
\newcommand{\E}{\mathbb{E}}
\newcommand{\Prb}{\mathbb{P}}
\newcommand{\wpralg}{\textsc{Weighted Pairwise Rounding}}
\newcommand{\scheduledalg}{\textsc{Scheduled Weighted Allocation}}
\newcommand{\binaryalg}{\textsc{Switched Binary Allocation}}

\newif\ifpolishmarks
\polishmarksfalse
\DeclareRobustCommand{\polish}[1]{\ifpolishmarks\leavevmode{\color{red}#1}\else#1\fi}
\newenvironment{proofof}[1]{{\vspace*{5pt} \noindent\bf Proof of #1:  }}{\hfill\rule{2mm}{2mm}\vspace*{5pt}}

\title{Randomized Algorithm for MMS Allocation of Chores\\
with Online Agent Arrivals}
\author{Haolong Li \and Fangxiao Wang \and Xiaowei Wu}
\date{\today}

\begin{document}
\maketitle

\begin{abstract}
We study the fair allocation of $m$ indivisible chores among $n$ agents in the agent-arrival model against an oblivious adversary. The chore set and the number of agents are known in advance. When agent $i$ arrives, her nonnegative additive cost function $c_i$ is revealed, and the algorithm must immediately and irrevocably assign her a subset of the unallocated chores. We focus on the notion of maximin share (MMS) and measure the fairness of the realized allocation $X$ by the ex-post MMS approximation ratio $\rho(X)$.
We measure an algorithm's performance by its expected ratio, and also the guarantees that hold with high probability.

For binary additive costs, we achieve an expected ex-post MMS ratio of $3/2+o(1)$, and show that every randomized online algorithm has an expected ratio of at least $3/2-o(1)$. For every $n$, the expected ratio is at most $1.755$. 
A variant of our algorithm guarantees an approximation ratio of $2$ with high probability, which is optimal.
For general additive costs, we present a weighted-rounding algorithm that achieves an expected ratio of $2+o(1)$, which also holds with high probability.
The expected ratio is at most $2.594$ for every $n$.
We also show that similar algorithms extend to additive goods, achieving a constant expected ex-post MMS ratio and guaranteeing every agent an $x$-fraction of her MMS with probability at least $1-85x^2$ for every $x\le1/15$.
\end{abstract}

\newpage
\section{Introduction}

Fair allocation studies how to distribute resources or tasks among agents with different preferences while ensuring that each agent receives a fair share. For indivisible items, achieving fairness is challenging because items cannot be divided and agents may value them differently. This difficulty is particularly apparent for chores, which impose costs such as time and effort on the agents who perform them. A widely studied fairness benchmark is the maximin share (MMS), originally introduced by Budish for goods \cite{budish2011combinatorial}. For chores, an agent's MMS is the minimum possible cost of her most costly bundle when she partitions all chores into $n$ bundles. An allocation is $\alpha$-MMS, for $\alpha\geq 1$, if every agent receives a bundle whose cost is at most $\alpha$ times her MMS. Although exact MMS allocations need not exist for additive chores, approximation guarantees have improved considerably. Aziz et al. obtained a $2$-approximation and a $5/3$-approximation using only ordinal preferences \cite{aziz2017algorithms,aziz2024ordinal}. Feige and Huang improved the ordinal guarantee to $8/5$ \cite{feige2023picking}, while Huang and Segal-Halevi obtained a $13/11$ guarantee using numerical cost information \cite{huang2023reduction}.

In many applications, agents or requests arrive gradually even when all items are available in advance. Consider an organizer who needs to distribute a fixed collection of tasks among volunteers who arrive one by one. When a volunteer arrives, the organizer learns her costs and must immediately assign her a bundle of tasks, without knowing how costly the remaining tasks will be for future volunteers. A similar challenge arises in disaster relief, where a relief center holds a fixed stock of food and medical supplies while requests from affected communities or emergency teams arrive over time. Supplies must be distributed promptly to meet urgent needs, without knowing what resources future requests will require. These applications motivate the online agent-arrival model for fair allocation \cite{kulkarni2025online,li2026arrivals}. In this model, the set of items and the number of agents are known in advance, but each agent's valuation or cost function is revealed only upon her arrival. The algorithm must immediately and irrevocably allocate a subset of the remaining items to each arriving agent, and all items must be allocated by the end. The central challenge is to guarantee fairness without knowing the preferences of agents who have not yet arrived.
In this work, we mainly focus on the case of chores, while some of our results extend to the case of goods.

Moving from offline allocation to the agent-arrival model requires making allocation decisions before future agents' preferences are known. Early studies of agent arrivals considered divisible resources, including cake cutting and the allocation of multiple resources \cite{walsh2011cake,kash2014dynamic}. For indivisible goods, Kulkarni et al. studied MMS allocation when the possible valuation functions are known in advance, but each agent's valuation is revealed only upon her arrival \cite{kulkarni2025online}. For chores, Li et al. investigated deterministic MMS guarantees with and without prior knowledge of the possible cost functions \cite{li2026arrivals}. With this prior information, they obtained logarithmic and constant approximation ratios for additive and binary costs, respectively. However, without prior cost information, they showed that no deterministic algorithm can guarantee an approximation factor below $n$, even for binary costs.
The lower bound above applies to deterministic algorithms, and Li et al. explicitly left open whether randomization can improve MMS guarantees \cite{li2026arrivals}. 

In this work, we study this question for both binary and general additive costs, without prior information about future agents' cost functions. We consider an oblivious adversary, which fixes all cost functions and the arrival order before the algorithm makes any random choices. For each realized allocation, we define its \emph{ex-post MMS ratio} as the maximum, over all agents, of the ratio of an agent's allocated cost to her MMS. We seek a constant upper bound on the expectation of this ratio, independent of the number of agents and chores. We also study high-probability guarantees: the ex-post MMS ratio should be bounded by a constant with probability tending to one as the number of agents grows.

\subsection{Our Contributions}

In this work, we show that constant approximation ratios can be achieved both in expectation and with high probability using randomized algorithms.

\paragraph{Result 1 (Binary additive costs).} We propose a randomized online algorithm for binary additive costs that achieves an expected MMS ratio of $3/2+o(1)$. Conversely, every randomized online algorithm has a fixed oblivious binary instance with an expected ratio of at least $3/2-o(1)$.

\medskip

Thus, the asymptotically optimal expected ratio in the binary domain is $3/2$. We also obtain an expected ratio below $1.755$ for every $n$. For a high-probability guarantee, we find that the expected $3/2+o(1)$ guarantee cannot be strengthened to hold with high probability: every randomized online algorithm has a fixed oblivious binary instance on which the realized ratio is at least $2$ with constant probability. We show that this threshold is tight by introducing a randomized online algorithm which has an ex-post MMS ratio of at most $2$ except with probability $e^{-\Omega(n)}$.

\paragraph{Result 2 (General additive costs).} 
We propose a randomized online algorithm for general additive costs with an expected MMS ratio of $2+o(1)$, which also holds with high probability. 

\medskip

We also propose a variant of our algorithm that guarantees an expected MMS ratio of at most $2.594$ for every $n$. Since binary costs are a special case of additive costs, the limiting ratio $2$ is optimal for high-probability guarantees in the general additive domain as well.

\smallskip

For both binary and general additive costs, we show that the expected ratios can be improved for small $n$.
The results are summarized in the following table.

\begin{table}[ht]
  \centering
\small \resizebox{0.85\textwidth}{!}{%
  \begin{tabular}{@{}c@{\qquad}cc@{\qquad}cc@{}}
    \toprule
    & \multicolumn{2}{c}{Binary additive costs} & \multicolumn{2}{c}{General additive costs} \\
    \cmidrule(lr){2-3}\cmidrule(lr){4-5}
    $n$ & Lower bound & Upper bound & Lower bound & Upper bound \\
    \midrule
    $2$ & $1.2500$ & $1.2500$ & $1.2500$ & $1.5000$ \\
    $3$ & $1.2962$ & $1.2985$ & $1.2962$ & $1.8889$ \\
    $4$ & $1.3164$ & $1.6722$ & $1.3164$ & $<2.2130$ \\
    General $n$ & $3/2-O(n^{-1/2})$ & $3/2+O(n^{-1/5})$ & $3/2-O(n^{-1/2})$ & $2+O(\sqrt{\log n/n})$ \\
    \bottomrule
  \end{tabular}}
\caption{Bounds on the optimal expected ratio.}
\end{table}

We also generalize our algorithms and analysis to the case of goods. 

\paragraph{Result 3 (General additive rewards).} 
We propose a randomized online algorithm for general additive valuations with an expected MMS ratio of $1/10$. We also present a hardness result of $0.5$ for the expected ex-post MMS ratio.

\medskip

Unfortunately, as shown by Chawla et al.~\cite{chawla2026singleorsampleonlinefairallocation}, no online algorithm can guarantee a constant MMS ratio with high probability.
Instead, we obtain a trade-off between the ratio lower bound and the failure probability:
For every $0< x \leq 1/15$, we can guarantee that the MMS ratio is at least $x$ with a probability at least $1-85\cdot x^2$. Choosing $x = 1/15$ yields a ratio of $1/15$ with probability $0.62$, improving the $1/24$ ratio with probability $0.423$ in \cite{chawla2026singleorsampleonlinefairallocation}.

\subsection{Overview of the Techniques}

We begin with binary additive costs, where each chore costs either zero or one. All known hard instances for agent-arrival MMS, both in previous work and in this paper, are binary. Understanding this case may therefore guide the design of better algorithms for general costs. It also permits cleaner allocation rules and probabilistic analysis. Throughout this overview, we discuss only chores and focus on the expected MMS ratio.

\paragraph{Binary additive costs: the need for randomization.}
The main difficulty is choosing which chores to allocate without knowing their costs to future agents. Suppose the first agent regards all $m$ chores as costly, while each later agent regards only the same hidden set of $n$ chores as costly. If the early allocations miss this set, too many costly chores may remain for later agents, even when the first agent's allocation appears fair. This lack of information also underlies the hardness of online allocation with arriving agents \cite{kulkarni2025online}. In particular, Li et al.~\cite{li2026arrivals} show that, without prior information about future costs, no deterministic algorithm achieves an MMS ratio below $n$, even for binary costs. Randomization may help allocate some hidden costly chores without knowing their identities.
A natural randomized rule gives each arriving agent all remaining zero-cost chores and a uniformly random subset of her remaining unit-cost chores. The key question is how many costly chores to allocate. Our lower bound shows that even randomized online algorithms cannot guarantee an expected ratio better than $3/2$ asymptotically. To approach this bound, one might assign each agent roughly $3\mu_i/2$ costly chores, where $\mu_i$ is her MMS. However, we show that this only guarantees an expected ratio below $1.755$. The difficulty is that $3\mu_i/2$ already imposes a ratio of $3/2$ in favorable outcomes, while unfavorable outcomes may leave too many chores for later agents and incur a substantially larger ratio. Thus, this rule is overly cautious about future workloads and loses the opportunity to achieve a smaller ratio in favorable outcomes.

\paragraph{Dynamic aggressiveness.}
Our algorithm instead varies its \emph{aggressiveness} over time. In aggressive rounds, the arriving agent receives fewer costly chores, keeping her own ratio small; in cautious rounds, she receives more, reducing the workload left for later agents. Before arrivals, we randomly choose the number $T$ of initial aggressive rounds. The algorithm uses the smaller allocation size in the first $T$ rounds and the larger one afterward, until the common costly set becomes small enough to distribute evenly. Randomizing the transition hedges against the unknown point at which future agents' common costly chores become apparent. A fixed transition would instead allow an oblivious adversary to exploit its timing.
Concentration inequalities are natural analytical tools, but applying them is not immediate. Random selections in different rounds are dependent: a chore allocated earlier cannot be selected again, and the available set changes after every allocation. We compare this process with an auxiliary process that samples independently, with replacement, from the relevant costly sets in each round. Coupling and convex-order comparisons relate the processes, allowing us to apply classical bounds, including Chernoff bounds and Hoeffding's comparison for sampling without replacement. This gives an expected ratio of $3/2+o(1)$.

\paragraph{General additive costs.}
Our general-additive algorithm follows the same idea of balancing the workload assigned now with the flexibility left for later agents. Since chores have different costs, fixing the number of allocated chores is no longer meaningful. Instead, before arrivals we fix a schedule $(p_1,\ldots,p_n)$ with $0<p_i\le1$ and $p_n=1$. On arrival, agent $i$ is fractionally assigned a $p_i$-share of every remaining chore, and hence a $p_i$-fraction of her total remaining cost. The choice $p_n=1$ ensures that all chores are allocated by the end.
Because chores are indivisible, these fractional assignments must be rounded. We want to preserve each chore's allocation probability and the current agent's expected cost while avoiding large fluctuations in her realized cost. Independent rounding preserves expectations, but the errors of many costly chores can accumulate, making a large cost for some agent more likely. We therefore introduce \emph{Weighted Pairwise Rounding}. Starting from a common fraction $p_i$, the procedure randomly adjusts two fractional coordinates at a time, preserving their expectations and the current agent's total fractional cost. Each adjustment makes at least one coordinate integral. Eventually, at most one fractional chore remains, and its final Bernoulli decision is the only source of rounding error in the current agent's total cost. Thus, her realized cost differs from the fractional target by at most one chore's cost. A product-moment inequality also controls the distribution of chores left for future agents.
Finally, concentration bounds compare the remaining chores with an independently retained set and express the expected MMS ratio in terms of the schedule $(p_1,\ldots,p_n)$. We choose the schedule to control the largest fractional workload across arrivals, obtaining an expected ratio of $2+o(1)$.

\subsection{Related Work}

There is an extensive literature on fair allocation, and we review only the results most closely related to our setting. For broader background, we refer to the survey on online fair division by Aleksandrov and Walsh~\cite{aleksandrov2020survey}, Aziz et al.~\cite{aziz2022survey}, and Amanatidis et al.~\cite{amanatidis2023survey}.

\paragraph{Offline MMS allocation.}
Budish introduced the MMS benchmark for goods \cite{budish2011combinatorial}. However, it is later shown that exact MMS allocations need not exist for additive goods or chores \cite{kurokawa2018fair,aziz2017algorithms}. For chores, approximation guarantees improved from $2$ to $4/3$, $11/9$, and $13/11$ \cite{aziz2017algorithms,barman2020approximation,huang2021framework,huang2023reduction}. Algorithms using only ordinal preferences achieve $5/3$ \cite{aziz2024ordinal} and, more recently, $8/5$ \cite{feige2023picking}. For goods, guarantees improved from $2/3$ to $3/4$ and beyond \cite{kurokawa2018fair,ghodsi2018fair,garg2021improved,akrami2023simplification,akrami2024breaking}, including a $10/13$ guarantee \cite{heidari2026improved} and a $7/9$ guarantee with an associated FPTAS \cite{huang2025sevenninths}. Impossibility bounds also constrain what can be achieved offline: Garbuz and Ezra established bounds of $20/21$ for goods and $31/30$ for chores \cite{garbuz2026impossibility}, improving the earlier $39/40$ and $44/43$ bounds \cite{feige2021negative}. These results assume that the preferences of all agents are available before allocation.

\paragraph{Online fair division.}
Online fair division requires irrevocable decisions before all relevant information has been revealed. Early work on arriving agents and divisible resources considered online cake cutting and dynamic multi-resource allocation \cite{walsh2011cake,kash2014dynamic}. Under item arrivals, Aleksandrov et al.~\cite{aleksandrov2015food} studied the sequential allocation of indivisible goods, motivated by food-bank distribution, and analyzed a randomized rule that is envy-free in expectation. Subsequent work explored the evolution of envy, the trade-off between fairness and efficiency, and reallocations needed to maintain envy-freeness up to one item \cite{benade2024efficient,he2019future}. Approximate proportionality under item arrivals was considered by Choo et al.~\cite{choo2025proportionality}, while Chen and Tan~\cite{chen2026competitive} examined competitive guarantees under several fairness notions for goods and chores. Long and Chen~\cite{long2026randomized} investigated randomized algorithms against an oblivious adversary, distinguishing bounds on expected realized fairness from high-probability guarantees. Unlike these fairness notions, the MMS benchmark compares each agent's realized allocation with a partition-based share of the entire item set.

\paragraph{Online MMS allocation.}
Under item arrivals, Zhou et al.~\cite{zhou2023online} obtained a $(2-1/n)$-MMS guarantee for additive chores when each agent's total cost is known; without this information, the tight deterministic ratio is $n$ \cite{seddighin2025lower,song2025online}. In contrast, exact MMS is attainable for binary costs \cite{wang2026binary}. Related results for goods consider restricted valuations and predictions about future items \cite{amanatidis2025personalized,neoh2025additional}. For randomized algorithms, Neoh and Teh~\cite{neoh2026closing} established lower bounds against adaptive adversaries, while Cohen et al.~\cite{cohen2026oblivious} gave an impossibility result for positive MMS guarantees with high probability against oblivious adversaries. Under agent arrivals, Kulkarni et al.~\cite{kulkarni2025online} studied indivisible goods when possible valuation functions are known beforehand, while Li et al.~\cite{li2026arrivals} investigated chores with and without prior cost-type information. Their deterministic lower bound of $n$ without such information holds even for binary costs.

\section{Preliminaries}

For any integer $t$, we use $[t]$ to denote $\{1,2,\ldots,t\}$. 
We define the maximum over an empty index set to be zero.
We write $x^+:=\max\{x,0\}$. Let $N=[n]$ be the set of agents and let $M$ be the set of $m$ indivisible chores. The set $M$ and the number of agents $n$ are known in advance.
Agents arrive in the order $1,2,\ldots,n$. When agent $i$ arrives, her nonnegative additive cost function $c_i:2^M\to\mathbb R_{\ge0}$ is revealed. Throughout, costs are additive, so $c_i(S)=\sum_{e\in S}c_i(e)$ for every $S\subseteq M$.
By adding dummy chores whose costs are fixed to zero for every agent, we may assume without loss of generality (w.l.o.g.) that $M$ contains arbitrarily many chores; in particular, we assume $m>n$.

\begin{definition}[Maximin Share]
Let $\Pi_n(M)$ be the set of ordered $n$-partitions of $M$. We use $\mu_i$ to denote the maximin share (MMS) of agent $i$:
    \begin{equation*}
      \mu_i:=\MMS_i(M,n)
      :=\min_{(P_1,\ldots,P_n)\in\Pi_n(M)}
           \max_{j\in[n]} \left\{c_i(P_j)\right\}.
    \end{equation*}
\end{definition}

For additive cost functions, computing the MMS value is NP-hard in general; our algorithms use $\mu_i$ only in the analysis.
We may assume that $\mu_i>0$ for every agent, because otherwise we can assign all remaining chores to this agent. The following properties of $\mu_i$ will be used repeatedly.

\begin{lemma}[Basic MMS properties]
\label{lem:mms-lb}
For every agent $i$, we have
\begin{equation*}
  \mu_i\ge \frac{c_i(M)}{n}
  \qquad\text{and}\qquad
  \mu_i\ge \max_{e\in M}\{c_i(e)\}.
\end{equation*}
\end{lemma}

Upon the arrival of agent $i$, the algorithm must immediately assign her a subset $X_i$ of the remaining chores. The decision is irrevocable. The output is a complete allocation $X=(X_1,\ldots,X_n)$: the bundles are pairwise disjoint and $\bigcup_{i\in N}X_i=M$. The arrival sequence is fixed by an \emph{oblivious adversary}: the full sequence $(c_1,\ldots,c_n)$ is chosen before the randomized online algorithm draws its random coins, although the cost functions are revealed only upon arrival.

\begin{definition}[Approximate MMS allocation]
For $\alpha\ge1$, a complete allocation $X=(X_1,\ldots,X_n)$ is \emph{$\alpha$-MMS} if for every $i\in N$ we have $c_i(X_i)\le\alpha\cdot \mu_i$.
\end{definition}

The algorithms we study are randomized, so the output allocation $X$ is random. We measure the realized MMS ratio of allocation $X$ as $\rho(X):=\max\left\{1,\max_{i\in [n]} \{\frac{c_i(X_i)}{\mu_i}\} \right\}$.

\begin{definition}[Competitiveness]
A randomized online algorithm is \emph{$\alpha$-competitive in expectation} if $\E[\rho(X)]\le\alpha$ for every online instance. It is \emph{$\alpha$-competitive with failure probability $\delta$} if $\Prb[\rho(X)>\alpha]\le\delta$.
We say that the guarantee holds \emph{with high probability} if $\delta=o(1)$ as $n\to\infty$.
\end{definition}

\section{Binary Additive Costs}

In this section, we consider the binary setting, where $c_i(e)\in\{0,1\}$ for every agent $i\in N$ and chore $e\in M$. For agent $i$, we define her \emph{support} as $B_i:=\{e\in M:c_i(e)=1\}$.
Notice that $\mu_i=\left\lceil\frac{|B_i|}{n}\right\rceil$ for every agent $i\in N$.
The binary setting allows a simple allocation rule: for each arriving agent $i$, we may first assign her all zero-cost chores and then a uniformly random item set of roughly $\alpha\cdot \mu_i$ remaining costly chores.
The parameter $\alpha$ controls how cautious the algorithm is and can be optimized to minimize the expected ratio.
The binary setting also reveals the core difficulty of the problem: without information about future agents, it is difficult to decide which chores to assign to the current agent, especially when the chores are indistinguishable.
In fact, as in related work, all our hard instances are binary.

\subsection{Two Simple Hard Instances}

We begin with two lower bounds that illustrate the main difficulty of the online model.

\begin{theorem}
For every $n\ge2$, no online algorithm can achieve an expected competitive ratio strictly better than $1+\left(1-\frac1n\right)^n$.
\end{theorem}
\begin{proof}
Fix a positive integer $s$ and let $M$ contain $n\cdot s$ chores. Agent~1 regards every chore as costly, so her MMS is $s$. Before the algorithm starts, choose a uniformly random item subset $C\subseteq M$ of size $n$, and let every later agent have support $C$ and MMS one.
Fix a deterministic algorithm and let $X_1\subseteq M$ be the bundle assigned to agent $1$. Let $q=|X_1|$. If $q\le s$, then the MMS approximation ratio of the first agent is $1$. However, if $X_1$ is disjoint from $C$, then $n$ costly chores remain for $n-1$ agents, forcing a competitive ratio of at least $2$. Therefore, the expected ratio of the algorithm is at least $1+{\binom{ns-q}{n}}/{\binom{ns}{n}}$.
Similarly, if $s<q<2s$, then the ratio of the first agent is $x=q/s$. The overall ratio is at least $x$ if $X_1\cap C\neq\varnothing$ and at least $2$ otherwise. Hence the expected ratio of the algorithm is at least
\begin{equation*}
  x+(2-x)\cdot {\binom{ns-q}{n}}/{\binom{ns}{n}}.
\end{equation*}

Note that $q\ge2s$ already gives a ratio at least two. As $s\to\infty$, these expressions converge uniformly to $1+(1-x/n)^n$ for $0\le x\le1$ and to $x+(2-x)(1-x/n)^n$ for $1\le x\le2$. The first expression decreases, and the derivative of the second is
\begin{equation*}
  1-\left(1-\frac xn\right)^{n-1}\left(3-x-\frac xn\right),
\end{equation*}
which is nonnegative on $[1,2]$. The common minimum is attained at $x=1$ and equals $1+(1-1/n)^n$. Since $C$ is chosen uniformly at random, the lower bound follows immediately by Yao's principle.
\end{proof}

For $n=2$, the construction gives $5/4$; for $n=3$, it gives $35/27\approx 1.296$. The hard instance suggests the importance of randomizing the bundle assigned to the first agent: on similar instances, deterministic algorithms can only guarantee an $\Omega(n)$-competitive ratio~\cite{li2026arrivals}.

As $n\to\infty$, the above hard instance gives a lower bound of $1+1/e\approx 1.368$. We next present a second construction that gives a lower bound of $1.5$ as $n\to \infty$.

\begin{theorem}\label{thm:binary-hard-instance}
Let $n\ge4$. For every randomized online algorithm, there exists a fixed binary instance on $2n$ chores for which
\begin{equation*}
  \E[\rho(X)] \geq \frac32-\frac{1+o(1)}{\sqrt{\pi n}}.
\end{equation*}
\end{theorem}
\begin{proof}
Set $t:=\lfloor n/2\rfloor$. Let $|M|=2n$ and let $C\subseteq M$ be a uniformly random item subset of size $n$. The first $t$ agents have support $M$ and MMS two, whereas agents $t+1,\ldots,n$ have support $C$ and MMS one.
Fix a deterministic online algorithm. If any of the first $t$ agents receives at least three chores, then the realized ratio is already at least $3/2$. We may therefore consider the complementary case, in which their union $U$ has size at most $2t$. Extend $U$ to an arbitrary fixed item set $U'\subseteq M$ of size $2t$. Since $U\subseteq U'$, we have $|C\cap U|\le |C\cap U'|$. Hence, on the event $|C\cap U'|<t$, fewer than $t$ chores of $C$ have been allocated to the first $t$ agents. More than $n-t$ chores of $C$ are then left for the remaining $n-t$ agents. Every one of these chores costs one to every remaining agent, so some remaining agent receives at least two costly chores, and the realized ratio is at least two.
Consequently,
$\E[\rho(X)] \ge 1+\Prb[|C\cap U'|<t]$.
The random variable $|C\cap U'|$ is hypergeometric with population size $2n$, $n$ marked elements, and sample size $2t$. Its distribution is symmetric about $t$. Therefore,
\begin{equation*}
  \Prb[|C\cap U'|<t]
  =\frac12-\frac12\cdot \Prb[|C\cap U'|=t]
  =\frac12-\frac12\cdot \frac{\binom nt^2}{\binom{2n}{2t}},
\end{equation*}
and thus
\begin{equation*}
  \E_C[\rho(X)]
  \ge \frac32-\frac12\frac{\binom nt^2}{\binom{2n}{2t}}
  = \frac32-\frac{1+o(1)}{\sqrt{\pi n}},
\end{equation*}
where the last inequality follows by Stirling's formula.

The same argument applies after conditioning on any random seed of a randomized algorithm. Averaging first over the seed and then over the random choice of $C$, some fixed item set $C$ satisfies the same lower bound. Hence the hard instance is fixed before the algorithm draws its random coins and is valid against an oblivious adversary.
\end{proof}

\paragraph{Observations and Ideas.}
As the example above suggests, the best possible expected ratio is asymptotically at least $1.5$. 
A natural approach is therefore to set $1.5$ as the target ratio, assign each arriving agent $i$ roughly $1.5\mu_i$ random costly chores in addition to all zero-cost chores, and hope that few chores remain when the last agent arrives.
Unfortunately, the expected-ratio bound obtained from this rule is $1.755$, which does not approach the asymptotically optimal ratio.
The main reason is that, in ``unlucky'' runs (which happens with constant probability), agent $n$ may incur a large MMS ratio, whereas in ``lucky'' runs, even if all chores are allocated early, the resulting approximation ratio remains $1.5$.
This can keep the expected ratio strictly larger than $1.5$.
To address this issue, we use an aggressive rule in the early stage of the algorithm, assigning each agent only $\mu_i$ costly chores, and then switch to a cautious rule if many chores remain.
We let $T\in\{0,1,\ldots,n\}$ denote the number of initial aggressive rounds. Thus, round $T+1$ is the first cautious round.
However, we notice that a fixed switching point may be vulnerable: if the core becomes small shortly after the cautious phase begins, agent $T+1$ already incurs a ratio of $3/2$ while still leaving too many chores for the remaining agents. The construction above illustrates why leaving even a few additional chores in a hidden support can raise a later agent's ratio to two.
We therefore draw $T$ from a carefully designed distribution, which may depend on $n$, before any agent arrives.
Randomizing $T$ makes the algorithm more robust against the oblivious adversary and gives an upper bound of $1.5+o(1)$ on the expected competitive ratio.

\subsection{The Online Algorithm}

Let $T\in\{0,1,\ldots,n\}$ be the number of initial aggressive rounds, drawn from a distribution $\pi$, which may depend on $n$ and will be fixed later.
We introduce the \emph{core}, the set of chores that are costly to every agent who has arrived. 
For each $i\in[n]$, let $C_i:=\bigcap_{j=1}^i B_j$, where $C_0:=M$.
A round has a \emph{large core} if $|C_i|>n$ and a \emph{small core} if $|C_i|\le n$.
Note that the cores depend only on the online instance and are independent of the algorithm's random choices.
The algorithm updates the core after each arrival and uses it to make allocation decisions.
See Algorithm~\ref{alg:binary-core-insurance}.
In particular, we use $\kappa_i = \lceil |C_i|/n \rceil \leq \mu_i$ as the benchmark allocation quota.\footnote{The same analysis applies if we replace $\kappa_i$ by $\mu_i$ in the algorithm. We use $\kappa_i$ instead of $\mu_i$ only for analytical convenience.
While in general we could have $\kappa_i < \mu_i$, as in the hard instances above, the instances follow a nested support structure, for which $C_i = B_i$ and thus $\kappa_i = \mu_i$.} 
Whenever the core is large, rounds $1,\ldots,T$ use the aggressive quota $\kappa_i$, and round $T+1$ is the first round that uses the cautious quota about $1.5 \kappa_i$.
Thus $T=0$ means that the algorithm is cautious from the beginning, whereas $T=n$ makes every round aggressive.
Once the core becomes small, it suffices to partition the remaining chores as evenly as possible.

\begin{definition}[Transition index]
For a fixed binary arrival sequence, define
$h:=\max\{i\in\{0,1,\ldots,n\}:|C_i|>n\}$.
\end{definition}

Thus rounds $1,\ldots,h$ have large cores. If $h<n$, round $h+1$ is the first round with a small core. The value of $h$ depends only on the cost functions and is independent of the algorithm's random choices.
We can assume w.l.o.g. that $h\geq 1$, as otherwise $|B_1|\leq n$, and the algorithm can guarantee an MMS allocation by assigning one costly item to every future agent.
The transition index $h$ naturally partitions the algorithm into $h$ large-core rounds and $n-h$ small-core rounds.
If $T<h$, then rounds $1,\ldots,T$ are aggressive (the arriving agent receives only $\kappa_i$ costly chores), while rounds $T+1,\ldots,h$ are cautious (the arriving agent receives about $1.5\kappa_i$ costly chores).
If $T\ge h$, all large-core rounds are aggressive.
The algorithm behaves differently in these phases, and different choices of $T$ yield different competitive-ratio guarantees.

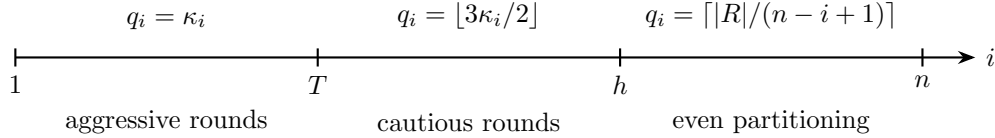
\begin{figure}[H]
  \centering
  \begin{tikzpicture}[x=1cm,y=1cm,>=Stealth,font=\small]
    \draw[->,line width=0.8pt] (0,0) -- (12.7,0) node[right] {$i$};
    \foreach \x in {0,4,8,12} {\draw[line width=0.8pt] (\x,0.12) -- (\x,-0.12);}
    \node[below=3pt] at (0,0) {$1$};
    \node[below=3pt] at (4,0) {$T$};
    \node[below=3pt] at (8,0) {$h$};
    \node[below=3pt] at (12,0) {$n$};
    \node[above=7pt] at (2,0) {$q_i=\kappa_i$};
    \node[above=7pt] at (6,0) {$q_i=\lfloor 3\kappa_i/2\rfloor$};
    \node[above=7pt] at (10,0) {$q_i=\left\lceil |R|/(n-i+1)\right\rceil$};
    \node at (2,-0.85) {aggressive rounds};
    \node at (6,-0.85) {cautious rounds};
    \node at (10,-0.85) {even partitioning};
  \end{tikzpicture}
  \caption{Quota schedule when $T<h<n$.}
  \label{fig:binary-quota-schedule}
\end{figure}

\begin{algorithm}[htp]
\caption{$\binaryalg(\pi)$}
\label{alg:binary-core-insurance}
\KwIn{$n$, $M$, a distribution $\pi$.}
\KwOut{A random complete allocation $X=(X_1,\ldots,X_n)$.}
Draw $T\sim\pi$. Set $R\leftarrow M$ and $C \gets M$\;
\For{$i=1,\ldots,n$}{
  Let $B_i$ be the reported costly-item set\;
  Compute $C\gets C\cap B_i$ and $\kappa_i \leftarrow\lceil|C|/n\rceil$\;
  Set $X_i \gets R\setminus B_i$ and $R\leftarrow R\cap B_i$\;
  Set the assignment quota to $q_i:=
          \begin{cases}
            \kappa_i,
              &|C| > n,\ i\le T,\\
            \lfloor3\kappa_i/2\rfloor,
              &|C| > n,\ i>T,\\
            \left\lceil \frac{|R|}{n-i+1}\right\rceil,
              &|C| \leq n.
          \end{cases}$ \\
  Let $Q_i\subseteq R$ be a uniformly random subset of size $\min\{|R|,q_i\}$\;
  Update $X_i \gets X_i\cup Q_i$ and set $R\leftarrow R\setminus Q_i$\;
  \If{$i=n$}{Give every chore still in $R$ to agent $n$ and set $R\leftarrow\varnothing$; \quad  \tcp{the final cleanup}}
}
\Return{$X$}\;
\end{algorithm}

For example, setting $T=n$ makes the algorithm aggressive in every large-core round.
The risk is that many chores may remain for the last agent, leading to a large approximation ratio.
As we will show, this choice can be preferable for small $n$ (see Theorem~\ref{thm:binary-all-agreesive}), because there are relatively few rounds in which errors can accumulate.
At the other extreme, setting $T=0$ makes the algorithm cautious from the beginning.
Under this rule, the probability that many chores remain for the last agent is very small.
This choice is particularly useful when seeking a competitive-ratio guarantee that holds with high probability (w.h.p.), as shown in the high-probability analysis below.
The disadvantage is that the resulting expected-ratio bound does not approach $1.5$ (see Theorem~\ref{thm:binary-immediate}).

\begin{theorem} \label{thm:binary-all-agreesive}
If $T=n$, then the expected ratios for $n=2,3,4$ satisfy, respectively,
\begin{equation*}
  \E[\rho(X)]\le 1.25,
  \qquad
  \E[\rho(X)]\le 1.2985,
  \qquad
  \E[\rho(X)]\le 1.6722.
\end{equation*}
\end{theorem}
\begin{proof}
The three direct calculations are given in
Appendix~\ref{app:binary-parameters}.
\end{proof}

\begin{theorem}\label{thm:binary-immediate}
If $T=0$, then for every $n$, there is a randomized online algorithm satisfying
\begin{equation*}
  \E[\rho(X)]<1.755.
\end{equation*}
\end{theorem}

To balance between these two extremes, we choose a random switching time $T$ and obtain an expected ratio approaching $3/2$ as $n$ grows.

\begin{theorem}
\label{thm:binary-upper}
Suppose that $n\ge16$, draw $T$ uniformly from $[\lfloor n^{2/5}\rfloor]$. Then
\begin{equation*}
  \E[\rho(X)]\le\frac32+O(n^{-1/5}).
\end{equation*}
\end{theorem}

\subsection{Independent-Sampling Comparison}

By the definition of the transition index $h$, for every large-core round $i\le h$, we have
\begin{equation*}
  \kappa_i:=\left\lceil\frac{|C_i|}{n}\right\rceil\geq 2,
  \quad
  \text{ and }
  \quad
  q_i = 
  \begin{cases}
    \kappa_i,&i\le T,\\
    \lfloor3\kappa_i/2\rfloor,&i>T.
  \end{cases}
\end{equation*}

Since $|C_i|\le n\cdot\kappa_i$ and $\kappa_i\ge2$, the base quota and the additional cautious quota satisfy
\begin{equation*}
  \kappa_i\cdot\frac{n}{|C_i|}\ge1,
  \qquad
  \left\lfloor\frac{\kappa_i}{2}\right\rfloor\cdot\frac{n}{|C_i|}\ge\frac{\lfloor\kappa_i/2\rfloor}{\kappa_i}\ge\frac13.
\end{equation*}

By definition, $C_{h+1}$ is the first core that is at most $n$.
Therefore, it is important to measure how many elements in $C_{h+1}$ have been allocated in the first $h$ rounds.
For ease of analysis, we fix a superset of $C_{h+1}$ of size $n$, and bound the number of allocated items in this superset.
In some sense, this allows us to assume w.l.o.g. that $|C_{h+1}|=n$.

\begin{definition}[Reference item set]
    When $1\le h<n$, fix an item set $\widehat C$ satisfying
    \begin{equation*}
      C_{h+1}\subseteq\widehat C\subseteq C_h,
      \qquad |\widehat C|=n.
    \end{equation*}
    We call $\widehat C$ the \emph{reference item set} and let $H$ be the number of chores of $\widehat C$ assigned during rounds $1,\ldots,h$. This definition also applies when $C_{h+1}=\varnothing$.
\end{definition}

\begin{lemma}\label{lem:binary-deterministic}
Before the final cleanup step in line 10 of Algorithm~\ref{alg:binary-core-insurance}, for every $i\leq h$, agent $i$ has an MMS approximation ratio at most one when $i\le T$ and at most $3/2$ when $i>T$. If $h<n$, every agent $i\ge h+1$ has ratio at most $\max\left\{1,\left\lceil\frac{n-H}{n-h}\right\rceil\right\}$.
\end{lemma}
\begin{proof}
After the zero-cost chores are assigned in round $i$, every remaining chore is costly to agent $i$. In a large-core round, the prescribed selection assigns her at most $\kappa_i$ such chores during the first $T$ rounds and at most $\lfloor3\kappa_i/2\rfloor$ from round $T+1$ onward. Since $\kappa_i\leq \mu_i$, her ratio is at most one or $3/2$, respectively.
If $h<n$, the algorithm enters the small-core phase in round $h+1$. Every remaining chore lies in $C_{h+1}\subseteq\widehat C$.
By definition, $H$ chores of $\widehat C$ have already been allocated. 
Thus, at most $n-H$ costly chores remain for the $n-h$ future agents. To prove the last statement, observe that agent $h+1$ receives at most $\lceil\frac{n-H}{n-h}\rceil$ costly chores. Moreover, because her quota is the ceiling of the average remaining load, at most $\lceil\frac{n-H}{n-h}\rceil\cdot(n-h-1)$ costly chores remain for the agents who have not arrived. The statement follows by induction.
\end{proof}

The key step in the analysis is to obtain a lower bound on $H$.
In fact, if we know the distribution of $H$, then we can fully characterize the distribution of $\rho(X)$.
Since $\widehat C$ is contained in every core $C_i$ for $i\leq h$, the random item subset allocated to agent $i$ may intersect $\widehat C$ and increase $H$.
However, selections across rounds are dependent: a chore selected in one round cannot be selected in a later round. This dependence makes the distribution of $H$ difficult to characterize.
We therefore introduce an auxiliary process with independent selections. A coupling argument is then used to show the stochastic dominance between the random processes and to provide the desired lower bound.

\paragraph{Independent-sampling Comparison.}
In every large-core round $i$, after the zero-cost chores have been removed, the algorithm selects distinct chores uniformly from the currently unallocated part of $C_i$. These selections are dependent across rounds because selected chores disappear and the cores are nested.
We compare them with the following auxiliary process. 
Consider the first $\ell$ large-core rounds and a nonempty item set $Q\subseteq C_\ell$. Conditional on the choice of $T$, let $q_i$ be the prescribed quota in round $i$, as defined in Algorithm~\ref{alg:binary-core-insurance}. Replace the actual selection in round $i$ with $q_i$ samples drawn independently and uniformly from the full core $C_i$, with replacement and independently across rounds. Each auxiliary sample in round $i$ hits $Q$ with probability $|Q|/|C_i|$. Let $K$ denote the total number of auxiliary sample occurrences that hit $Q$. Then
\begin{equation*}
  K\ \overset{d}{=}\ \sum_{i=1}^{\ell}\operatorname{Bin}\left(q_i,\frac{|Q|}{|C_i|}\right),
  \qquad
  \E[K]=\sum_{i=1}^{\ell}q_i\cdot\frac{|Q|}{|C_i|},
\end{equation*}
where the binomial variables are independent. See the figure below for an illustration. Notice that $K$ counts repeated occurrences: the same chore may contribute more than once. 

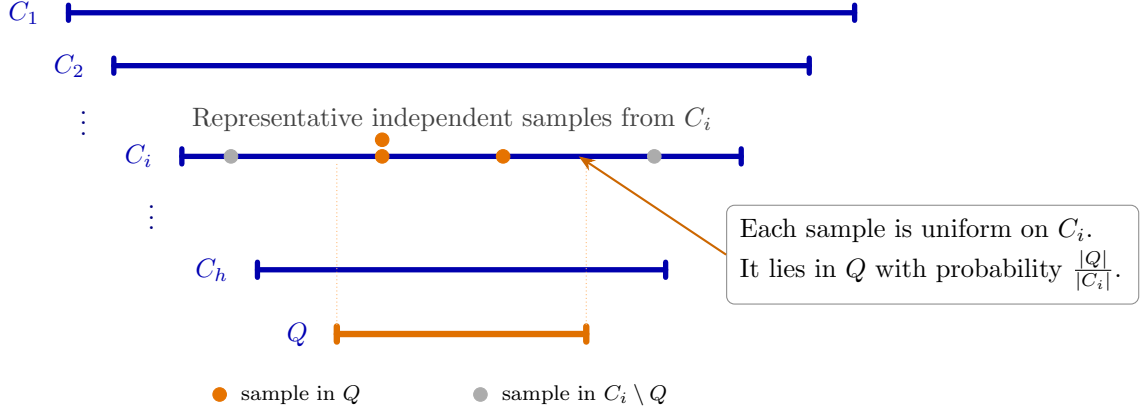
\begin{figure}
    \centering
    \begin{tikzpicture}[
      font=\small,
      interval/.style={line width=2.0pt, line cap=round, blue!68!black},
      omitted/.style={line width=1.4pt, line cap=round, blue!45!black, dashed},
      reference/.style={line width=2.4pt, line cap=round, orange!88!black},
      endpoint/.style={line width=0.9pt},
      samplein/.style={circle, fill=orange!90!black, draw=orange!90!black, inner sep=1.8pt},
      sampleout/.style={circle, fill=gray!65, draw=gray!65, inner sep=1.8pt},
      note/.style={draw=black!45, rounded corners=3pt, fill=white, align=left, inner sep=5pt},
      arrow/.style={-{Stealth[length=2.2mm]}, line width=0.8pt}
    ]
    
    \newcommand{\drawinterval}[5]{%
      \draw[#5] (#2,#4) -- (#3,#4);
      \draw[endpoint,#5] (#2,#4-0.10) -- (#2,#4+0.10);
      \draw[endpoint,#5] (#3,#4-0.10) -- (#3,#4+0.10);
      \node[anchor=east, text=blue!70!black] at (#2-0.25,#4) {#1};
    }
    
    \drawinterval{$C_1$}{-5.2}{5.2}{2.7}{interval}
    \drawinterval{$C_2$}{-4.6}{4.6}{2.0}{interval}
    \node[text=blue!50!black] at (-5.00,1.35) {$\vdots$};
    \drawinterval{$C_i$}{-3.7}{3.7}{0.8}{interval}
    \node[text=blue!50!black] at (-4.10,0.10) {$\vdots$};
    \drawinterval{$C_h$}{-2.7}{2.7}{-0.7}{interval}
    \drawinterval{$Q$}{-1.65}{1.65}{-1.55}{reference}
    
    \draw[orange!55, densely dotted] (-1.65,-1.43) -- (-1.65,0.70);
    \draw[orange!55, densely dotted] ( 1.65,-1.43) -- ( 1.65,0.70);
    
    \node[anchor=south west, text=black!75] at (-3.7,1.00) {Representative independent samples from $C_i$};
    \node[sampleout] at (-3.05,0.80) {};
    \node[samplein]  at (-1.05,0.80) {};
    \node[samplein]  at (-1.05,1.02) {}; 
    \node[samplein]  at ( 0.55,0.80) {};
    \node[sampleout] at ( 2.55,0.80) {};
    
    \node[note, right=1.00cm of {$(2.5,-0.5)$}, text width=5.1cm] (prob) {
      Each sample is uniform on $C_i$.\\[2pt]
      It lies in $Q$ with probability $\frac{|Q|}{|C_i|}$.
    };
    \draw[arrow, orange!80!black] (prob.west) -- (1.55,0.80);
    
    \node[samplein] at (-3.20,-2.35) {};
    \node[anchor=west, font=\scriptsize] at (-3.05,-2.35) {sample in $Q$};
    \node[sampleout] at (0.25,-2.35) {};
    \node[anchor=west, font=\scriptsize] at (0.40,-2.35) {sample in $C_i\setminus Q$};
    \end{tikzpicture}
    \caption{Illustration of the auxiliary random process.}
\end{figure}

The technical comparison between the actual without-replacement process and this auxiliary process is the key to the analysis.
Our analysis centers around a few key events.
For example, the event $T \geq h$ implies that the first $h$ rounds are all aggressive; the event $T < h$ implies that there are $T$ aggressive rounds and $h-T$ cautious rounds.
Another important event is $H \geq h$.
If this happens, then every agent $i\geq h+1$ has ratio at most $1$ (by Lemma~\ref{lem:binary-deterministic}), and we have $\rho(X)\leq 3/2$.
Therefore, a large part of our analysis focuses on upper bounding the probability of the event that $H < h$.

\begin{lemma}
\label{lem:binary-base}
Suppose that $h<n$, we have
\begin{equation*}
  \Prb[H<h\mid T\geq h]\le\frac12+\frac{h\cdot(h-1)}{2n}.
\end{equation*}
\end{lemma}
\begin{proof}
Condition on $T\ge h$ and apply the independent-sampling framework with $\ell=h$ and $Q=\widehat C$. Since every round $i\le h$ is aggressive, we have $q_i=\kappa_i$, and hence
\begin{equation*}
  \E[K]=\sum_{i=1}^{h}\kappa_i\cdot\frac{n}{|C_i|}\ge h.
\end{equation*}
For the coupling used here, let $G_i$ be the item set of currently unallocated chores that occur among the auxiliary samples in round $i$. Conditional on its size and the preceding history, $G_i$ is a uniformly random item subset of the unallocated part of $C_i$, and $|G_i|\le\kappa_i$. Complete $G_i$ uniformly to the prescribed selection size, or take the entire unallocated pool if it contains fewer than $\kappa_i$ chores. The resulting selection is uniform of the required size and contains every currently unallocated sampled chore. Under this coupling, every distinct chore of $\widehat C$ hit by the auxiliary process is allocated by the algorithm, either earlier or in the round when it is sampled.

Reducing its success probabilities to obtain a mean exactly $h$ gives a stochastically smaller sum $K_0$. The integer-mean median theorem of Jogdeo and Samuels~\citep{jogdeo1968monotone} states that an integer mean of a sum of independent Bernoulli variables is a median. Hence
$\Prb[K<h]\le \Prb[K_0<h]\le\frac12$.
Moreover,
\begin{equation*}
  \Prb[H<h\mid T=t]
  \le \Prb[K<h]+\Prb[H<h,\ K\ge h\mid T=t].
\end{equation*}
It remains to bound the second term.
Conditional on $K\ge h$, the identities of the first $h$ auxiliary occurrences in $\widehat C$ are independent and uniform over its $n$ chores. If these $h$ chores are all distinct, then the coupling gives $H\ge h$. Therefore, on the event $K\ge h$, the event $H<h$ requires at least one repeated pair among these first $h$ occurrences. Each fixed pair coincides with probability $1/n$, so a union bound over the $\binom h2$ pairs gives
\begin{equation*}
  \Prb[H<h,\ K\ge h\mid T=t]\le\frac{h\cdot(h-1)}{2n}.
\end{equation*}
Combining the two bounds proves the lemma.
\end{proof}

The same independent-sampling framework also gives the next two consequences. Let $H_Q$ denote the number of chores of $Q$ selected by the actual algorithm in the first $\ell$ rounds.

\begin{lemma}[Convex comparison]
\label{lem:binary-convex}
Under the setup above, let $\phi$ be a nonnegative, nonincreasing, convex function on the nonnegative integers. Then
\begin{equation*}
  \E\left[\phi(H_Q)\cdot\mathbf 1_{\{H_Q<|Q|\}}\right]
  \le
  \E[\phi(K)]
  \le
  \E[\phi(\operatorname{Pois}(\E[K]))].
\end{equation*}
In particular, if $\phi(|Q|)=0$, then the indicator in the first expectation may be omitted.
\end{lemma}

The proof of Lemma~\ref{lem:binary-convex} uses standard coupling arguments and is deferred to Appendix~\ref{app:binary}.
To express the lower-tail estimate in a form that will also be used for general additive costs, define
\begin{equation*}
    \psi(x):=x\log x-x+1, \qquad \forall x>0.
\end{equation*}
We have $\psi(1)=0$, $\psi'(x)=\log x$, and
$\psi''(x)=1/x>0$. 
It can be further verified (e.g., by taking derivatives) that $\psi(x)\ge(1-x)^2/2$ for $0<x\le1$. 

\begin{corollary}
\label{cor:binary-selection-tail}
Under the setup above, for every $0<\lambda\le\E[K]$ and
$0\le y\le\min\{|Q|,\lambda\}$,
\begin{equation*}
  \Prb[H_Q<y]
  \le
  \exp\left[-\lambda\cdot \psi\left(\frac{y}{\lambda}\right)\right]
  \le
  \exp\left(-\frac{(\lambda-y)^2}{2\lambda}\right).
\end{equation*}
\end{corollary}
\begin{proof}[Proof sketch]
Apply Lemma~\ref{lem:binary-convex} with the exponential payoff $\phi(w)=e^{-\theta w}$, bound the resulting auxiliary exponential moment by its Poisson counterpart, and optimize at $\theta=\log(\lambda/y)$. The second inequality follows from the elementary bound $\psi(x)\ge(1-x)^2/2$ for $0<x\le1$; the complete argument is given in Appendix~\ref{app:binary}.
\end{proof}

\paragraph{Remark.}
The three lemmas/corollary serve different purposes despite using the same auxiliary process. Lemma~\ref{lem:binary-base} handles the all-aggressive regime in which $\E[K]$ may be only $h$. In this regime, Corollary~\ref{cor:binary-selection-tail} gives no useful bound at $y=h$, so we instead use the median of $K$ and control repeated auxiliary hits. Lemma~\ref{lem:binary-convex} applies to general quotas and target sets, and is useful when the objective is a convex decreasing loss. Corollary~\ref{cor:binary-selection-tail} uses the same independent-sampling comparison with an exponential terminal payoff and is effective when $\E[K]$ is separated from the bad threshold. 

\medskip

\subsection{Expected-Ratio Guarantees}

We now apply the comparison developed in the preceding subsection to the two choices of the switching time used in Theorems~\ref{thm:binary-immediate} and~\ref{thm:binary-upper}.

Next, we use Corollary~\ref{cor:binary-selection-tail} to upper bound the probability of $H$ being small when some of the first $h$ rounds are cautious.

\begin{lemma}\label{lem:binary-insurance-tail}
If $h<n$, then for any $0\le t\le h$ we have
\begin{equation*}
  \Prb[H<h\mid T=t]\le\exp\left(-\frac{(h-t)^2}{24h}\right).
\end{equation*}
\end{lemma}
\begin{proof}
Condition on $T=t$ and use the setup of Corollary~\ref{cor:binary-selection-tail} with $\ell=h$ and $Q=\widehat C$.
Since $\widehat C\subseteq C_h\subseteq C_i$ for every $i\le h$ and $|\widehat C|=n$, the variable $H_Q$ in Corollary~\ref{cor:binary-selection-tail} is exactly the reference count $H$. For the corresponding auxiliary variable $K$,
\begin{equation*}
  \E[K]  =
  \sum_{i=1}^{t}\kappa_i\cdot\frac{n}{|C_i|}
  +  \sum_{i=t+1}^{h}\left\lfloor\frac{3\kappa_i}{2}\right\rfloor\cdot\frac{n}{|C_i|}
  \ge  h+\frac{h-t}{3}.
\end{equation*}

Set $\lambda:=h+\frac{h-t}{3}$.
Then $0<\lambda\le\E[K]$ and $\lambda\le4h/3$. Also, because $h<n=|\widehat C|$, the threshold $y=h$ satisfies $y\le\min\{|\widehat C|,\lambda\}$. The corollary gives
\begin{equation*}
  \Prb[H<h\mid T=t]
  \le
  \exp\left(-\frac{((h-t)/3)^2}{2\lambda}\right)
  \le
  \exp\left(-\frac{(h-t)^2}{24h}\right).
  \qedhere
\end{equation*}
\end{proof}

The preceding tail estimate also gives a single bound for the event that the realized ratio exceeds two. We will use it both for the expected-ratio analysis and for the high-probability guarantee.

\begin{lemma}
\label{lem:binary-failure-two}
Suppose that the switching time satisfies $T\le L<n/2$. Then
\begin{equation*}
  \Prb[\rho(X)>2]
  \le
  \exp\left(-\frac{(n-2L)^2}{48n}\right).
\end{equation*}
\end{lemma}
\begin{proof}[Proof sketch]
When $h<n$, Lemma~\ref{lem:binary-deterministic} shows that a ratio above two is possible only when $h>n/2$ and too few reference chores have been removed; Lemma~\ref{lem:binary-insurance-tail} then gives the stated exponent. When $h=n$, we instead apply Corollary~\ref{cor:binary-selection-tail} to $Q=C_n$ to show that, except with the same probability bound, every chore in the final core is selected before cleanup; the complete case analysis is given in Appendix~\ref{app:binary}.
\end{proof}

Finally, we prove Theorem~\ref{thm:binary-immediate} and~\ref{thm:binary-upper}.

\begin{proofof}{Theorem~\ref{thm:binary-immediate}}
Recall that $T=0$, so every large-core round is cautious. 
First suppose that $h<n$. By Lemma~\ref{lem:binary-deterministic}, $\rho(X)\le \max\left\{\frac32,1+\left\lceil\frac{\left(h-H\right)^+}{n-h}\right\rceil\right\}$.
If $H\ge h$, the second term is one and hence $\rho(X)\le3/2$. If $H<h$, then $h-H$ is a positive integer, and since $n-h\ge1$,
\begin{equation*}
  1+\left\lceil\frac{h-H}{n-h}\right\rceil
  \le1+h-H
  =\frac32+\left(h-\frac12-H\right).
\end{equation*}
Thus, in either case, we have
$\rho(X)\le\frac32+\left(h-\frac12-H\right)^+$.

We now apply Lemma~\ref{lem:binary-convex} with $\ell=h$ and $Q=\widehat C$. Since $T=0$, the prescribed quota is $q_i=\left\lfloor\frac{3\kappa_i}{2}\right\rfloor$ for every $i\in[h]$.
Because $\widehat C\subseteq C_h\subseteq C_i$ for every $i\le h$, this is a valid choice of $Q$, and the variable $H_Q$ in Lemma~\ref{lem:binary-convex} is exactly $H$. Since $|\widehat C|=n$, the auxiliary variable $K$ satisfies
\begin{equation*}
  \E[K]
  =   \sum_{i=1}^{h}
  \left\lfloor\frac{3\kappa_i}{2}\right\rfloor
  \cdot\frac{n}{|C_i|} \geq \frac{4h}{3}.
\end{equation*}

Let $\phi(w):=\left(h-\frac12-w\right)^+$.
The function $\phi$ is nonnegative, nonincreasing, and convex, and since $h<n=|\widehat C|$, it satisfies $\phi(|\widehat C|)=0$.
Therefore, Lemma~\ref{lem:binary-convex} gives $\E[\phi(H)] \le \E[\phi(\operatorname{Pois}(\E[K]))]$.
A Poisson random variable is stochastically increasing in its mean, whereas $\phi$ is nonincreasing. Since $\E[K]\ge4h/3$,
\begin{equation*}
  \E[\phi(H)]
  \le
  \E\left[\left(h-\frac12-\operatorname{Pois}\left(\frac{4h}{3}\right)\right)^+\right]
  =s_h.
\end{equation*}
Therefore we have
\begin{equation*}
  \E[\rho(X)]
  \le \frac32 + \E[\phi(H)]
  \le \frac32+s_h
  \le \frac32+\max_{0\le k\le n-1}\{s_k\}.
\end{equation*}

It remains to consider $h=n$. In this case, every round has a large core. Partition $C_n$ into exactly $\kappa_n$ nonempty blocks $Q_1,\ldots,Q_{\kappa_n}$, each of size at most $n$. Such a partition exists because $|C_n|\le n\cdot\kappa_n$. For each block $Q_b$, let $H_b$ be the number of its chores selected by the prescribed quotas in the first $n-1$ rounds.
Every one of the first $n-1$ agents has a ratio of at most $3/2$. At the arrival of agent $n$, after her zero-cost chores are removed, every costly chore that remains belongs to $C_n$. Thus agent $n$ receives exactly the chores of $C_n$ not selected during the first $n-1$ rounds, including any chores assigned in the final cleanup. Since $\mu_n\ge\kappa_n$, we have
\begin{equation*}
  \rho(X) \le \max\left\{    \frac32,\frac1{\kappa_n}\sum_{b=1}^{\kappa_n}\left(|Q_b|-H_b\right)
  \right\}
  \le \frac32+\frac1{\kappa_n}
  \sum_{b=1}^{\kappa_n}
  \left(|Q_b|-\frac32-H_b\right)^+.
\end{equation*}

Fix a block $Q_b$. Apply Lemma~\ref{lem:binary-convex} with $\ell=n-1$ and $Q=Q_b$.
Since $T=0$, the prescribed quota is $q_i=\left\lfloor\frac{3\kappa_i}{2}\right\rfloor$ for every $i\in[n-1]$.
This is valid because $Q_b\subseteq C_n\subseteq C_i$ for every $i\le n-1$, and the variable $H_Q$ is exactly $H_b$. For the corresponding auxiliary variable $K$,
\begin{equation*}
  \E[K] =\sum_{i=1}^{n-1}
  \left\lfloor\frac{3\kappa_i}{2}\right\rfloor
  \cdot\frac{|Q_b|}{|C_i|}
  \ge \frac{4(n-1)|Q_b|}{3n}
  \ge \frac{4(|Q_b|-1)}{3},
\end{equation*}
where the first inequality uses the same $4/3$ contribution per cautious round after scaling by $|Q_b|/n$, and the second uses $|Q_b|\le n$.
Now take $\phi_b(w):=\left(|Q_b|-\frac32-w\right)^+$.
This function is nonnegative, nonincreasing, and convex, and it vanishes at $w=|Q_b|$. Therefore, Lemma~\ref{lem:binary-convex} and monotonicity in the Poisson mean give
\begin{align*}
  \E\left[\left(|Q_b|-\frac32-H_b\right)^+\right]
  &\le
  \E\left[\left(|Q_b|-\frac32-\operatorname{Pois}(\E[K])\right)^+\right]\\
  &\le
  \E\left[\left(|Q_b|-\frac32-\operatorname{Pois}\left(\frac{4(|Q_b|-1)}{3}\right)\right)^+\right]
  =s_{|Q_b|-1}.
\end{align*}
Since $|Q_b|\le n$, we have $0\le |Q_b|-1\le n-1$. 
Therefore we have $\E[\rho(X)] \le \frac32+\max_{0\le k\le n-1}\{s_k\}$.
This proves the first statement in all cases. Appendix~\ref{app:binary-parameters} verifies that $\max_{k\ge0}s_k=s_6=22781/(30e^8)$, which proves the second statement.
\end{proofof}

\begin{proofof}{Theorem~\ref{thm:binary-upper}}
Fix the transition index $h$. Since $\rho(X)\le n$ (by Lemma~\ref{lem:mms-lb}),
\begin{equation*}
  \E[\rho(X)]
  =\E[\min\{\rho(X),2\}]+\E[(\rho(X)-2)^+]
  \le \E[\min\{\rho(X),2\}]+n\cdot\Prb[\rho(X)>2].
\end{equation*}
Since $T$ is uniform on $[L]$, we have $T\le L$ almost surely. The assumption $L\le n/4$ and Lemma~\ref{lem:binary-failure-two} therefore give
\begin{equation*}
  \Prb[\rho(X)>2]
  \le
  \exp\left(-\frac{(n-2L)^2}{48n}\right)
  \le e^{-n/192}.
\end{equation*}

It remains to bound $\E[\min\{\rho(X),2\}]$. If $h=n$, we use one direct consequence of the selection-tail bound. Conditional on $T=t\le L\le n/4$, the auxiliary mean for $Q=C_n$ satisfies
\begin{equation*}
  \E[K]
  \ge
  |C_n|\left(1+\frac{n-t}{3n}\right)
  \ge\frac54\cdot|C_n|.
\end{equation*}
Taking $\lambda=5|C_n|/4$ and $y=|C_n|$ in Corollary~\ref{cor:binary-selection-tail} gives
\begin{equation*}
  \Prb[H_{C_n}<|C_n|\mid T=t]
  \le e^{-|C_n|/40}
  \le e^{-n/40}.
\end{equation*}
If every chore of $C_n$ has been selected before the final cleanup, every agent has ratio at most $3/2$; otherwise we use the truncation at two. Hence
\begin{equation*}
  \E[\min\{\rho(X),2\}]
  \le\frac32+\frac12\cdot e^{-n/40},
\end{equation*}
which is smaller than the claimed bound. We therefore assume $h<n$.
If $H\ge h$, Lemma~\ref{lem:binary-deterministic} gives $\rho(X)\le1$ when $T\ge h$ and $\rho(X)\le3/2$ when $T<h$. If $H<h$, we simply use $\min\{\rho(X),2\}\le2$. 
Therefore, we have the following upper bound
\begin{equation*}
  \E[\min\{\rho(X),2\}]
  \le 1+\frac12\cdot\Prb[T<h]+\Prb[T\ge h,\ H<h]+\frac12\cdot\Prb[T<h,\ H<h].
\end{equation*}

If $h>L$, the event $T\ge h$ is empty. Otherwise, Lemma~\ref{lem:binary-base} gives
\begin{equation*}
  \Prb[T\ge h,\ H<h]
  \le\left(\frac12+\frac{L\cdot(L-1)}{2n}\right)\cdot \Prb[T\ge h].
\end{equation*}

In either case, we have
\begin{equation*}
  \E[\min\{\rho(X),2\}]
  \le\frac32+\frac{L\cdot(L-1)}{2n}+\frac12\cdot\Prb[T<h,\ H<h].
\end{equation*}

Since $T$ is uniform on $[L]$, Lemma~\ref{lem:binary-insurance-tail} gives
\begin{equation*}
  \Prb[T<h,\ H<h]
  \le\frac1L\sum_{t=1}^{\min\{h-1,L\}}
  \exp\left(-\frac{(h-t)^2}{24h}\right)
  \le\frac1L\cdot \sum_{t=1}^{h-1}\exp\left(-\frac{(h-t)^2}{24(L+h-t)}\right).
\end{equation*}

For $L\ge1$ and $t\le h-1$, we have $\frac{(h-t)^2}{L+h-t} \ge \frac{h-t}{\sqrt L}-1$. Thus,
\begin{equation*}
  \Prb[T<h,\ H<h]
  \le\frac{e^{1/24}}{L\cdot\left(e^{1/(24\sqrt L)}-1\right)}
  \le\frac{24e^{1/24}}{\sqrt L}
  <\frac{26}{\sqrt L}.
\end{equation*}
Therefore,
\begin{equation*}
  \E[\min\{\rho(X),2\}]
  \le\frac32+\frac{L\cdot(L-1)}{2n}+\frac{13}{\sqrt L}.
\end{equation*}
Combining this bound with the estimate on $\Prb[\rho(X)>2]$ proves the claimed bound.
Finally, notice that for $n\ge16$, the choice $L=\lfloor n^{2/5}\rfloor$ satisfies $L\le n/4$. Both $L^2/n$ and $1/\sqrt L$ are $O(n^{-1/5})$, which proves the asymptotic bound.
\end{proofof}

\subsection{High-Probability Ratio}

\polish{The expectation result raises the question of whether a ratio of $3/2+o(1)$ can also be guaranteed with high probability.} The next theorem shows that this strengthening is impossible. In fact, every randomized online algorithm incurs a ratio of at least two with constant probability.

\begin{theorem}
For every randomized online algorithm and every $n\ge3$, there is a fixed oblivious binary instance on $2n$ chores such that
\begin{equation*}
  \Prb[\rho(X)\ge2]\ge\frac{n-2}{4(2n-1)}=\frac18-o(1).
\end{equation*}
\end{theorem}

\begin{proof}
Agent~1 has support $M$ and MMS two. Before the algorithm starts, choose a uniformly random item subset $C\subseteq M$ of size $n$, and let every later agent have support $C$ and MMS one. Fix a deterministic algorithm and let $Q$ be the first bundle. If $|Q|\ge4$, the first ratio is at least two. Otherwise, on the event $Q\cap C=\varnothing$, all $n$ chores of $C$ remain for $n-1$ agents, so some later agent has ratio at least two. Moreover,
\begin{equation*}
  \Prb[Q\cap C=\varnothing]\ge\frac{\binom{2n-3}{n}}{\binom{2n}{n}}=\frac{n-2}{4(2n-1)}.
\end{equation*}
Condition on the random seed and average over $C$ to obtain a fixed oblivious instance for every randomized algorithm.
\end{proof}

Thus no ratio strictly below two can hold with high probability. The matching upper bound does not need a random switching time. We use the \emph{always-cautious} version of Algorithm~\ref{alg:binary-core-insurance}, obtained by fixing $T=0$. Thus, in every large-core round, agent $i$ receives a uniformly random item set of up to $\lfloor3\kappa_i/2\rfloor$ costly chores. The algorithm still uses only the revealed supports and the current core; $T$ no longer appears as a random choice.

\begin{theorem}
For every fixed binary additive instance, the always-cautious algorithm satisfies
$\Prb[\rho(X)>2]\le e^{-n/48}=e^{-\Omega(n)}$.
\end{theorem}

\begin{proof}
For the always-cautious algorithm, $T=0$ deterministically. Apply Lemma~\ref{lem:binary-failure-two} with $L=0$. Its assumptions are satisfied because $0<n/2$, and it gives $\Prb[\rho(X)>2] \le e^{-n/48}$.
\end{proof}

\section{General Additive Costs}\label{sec:additive}

We now consider arbitrary nonnegative additive costs. As in the preliminaries, we may restrict attention to agents with positive MMS. We also assume without loss of generality that every individual chore has strictly positive cost by adding infinitesimal positive perturbations to zero-cost items.
For the analysis in this section, we normalize each agent's cost function so that her MMS equals one. By Lemma~\ref{lem:mms-lb}, this normalization implies $c_i(M)\le n$ and $c_i(e)\le1$ for every agent $i$ and chore $e$. This is only an analytical normalization: the online algorithm does not need to compute the MMS values. Under this normalization, $\rho(X)=\max\left\{1,\max_{i\in[n]}\{c_i(X_i)\}\right\}$.

Our algorithm first fixes a \emph{schedule} that specifies what fraction of the remaining chores each arriving agent should receive fractionally, and then rounds this fractional allocation to an integral bundle. As in the binary analysis, allocations across rounds are dependent because a chore allocated in one round disappears from all later rounds. We again avoid describing the exact joint distribution. Instead, the analysis compares the remaining item set with an independent benchmark whose exponential moments are explicit.
The main result of this section is showing that there exists an algorithm with ratio $2+o(1)$, which holds both in expectation and with high probability.

\begin{theorem}\label{thm:additive-high-probability}
For the allocation of indivisible chores to online agents with general additive cost functions, there exists an online algorithm satisfying
\begin{equation*}
  \E[\rho(X)]
  \le2+(4+o(1))\sqrt{\frac{\log n}{n}}.
\end{equation*}
Moreover, the same ratio holds with probability at least $1-1/n$.
\end{theorem}

The next subsection describes the algorithm for an arbitrary schedule and the weighted rounding rule. Subsection~\ref{sec:additive-load-analysis} analyzes this framework without committing to a particular choice of the schedule. Subsection~\ref{sec:additive-schedule-choice} then chooses the schedule using the parameters $\beta_i$ and $\hat\beta_i$ and proves both the expectation and the high-probability bounds.

\subsection{The Online Algorithm}\label{sec:additive-algorithm}

We index rounds by the arriving agent, so round $i$ is the round in which agent $i$ arrives, for $i=1,\ldots,n$. Let $R_i$ denote the item set of unallocated chores immediately before round $i$, so $R_1=M$.

\begin{definition}[Schedule]\label{def:schedule}
A schedule is a vector $(p_1,\ldots,p_n)$ with $0<p_i\le1$ for every $i\in[n]$ and $p_n=1$. In round $i$, the algorithm fractionally assigns a $p_i$-fraction of every chore in $R_i$ to agent $i$.
\end{definition}

Thus $p_i$ is the fraction of every remaining chore that the current agent should receive before rounding. The schedule is fixed before the agents arrive and does not depend on their reported cost functions. Conditional on the remaining item set $R_i$, the fractional cost assigned to agent $i$ is $p_i\cdot c_i(R_i)$. The item set $R_i$ is random because of the allocations in earlier rounds, so this fractional cost is also random.

\paragraph{Rounding the fractional allocation.}
A naive approach is to round every remaining chore independently, assigning each chore to the current agent with probability $p_i$. This preserves the fractional cost in expectation, but it can introduce unnecessary variance in the realized cost. Since our objective takes the maximum over all $n$ agents, a per-agent failure probability that is merely small is not enough: after a union bound over all agents, the resulting failure probability may be close to $1$. We therefore want the rounding rule to have two properties. First, it should preserve the current agent's expected cost while keeping her realized cost very close to the fractional cost, e.g., up to one item.
Second, the random item set of chores left unallocated should remain sufficiently well-behaved to control the costs of future agents.

\paragraph{Weighted Pairwise Rounding.}
The main idea is to move fractional mass between two chores at a time. Starting from $x_e=p$ for every chore $e\in R$, we choose two fractional coordinates $x_e,x_f$ and randomly increase one while decreasing the other. The two changes are scaled by $c_i(e)$ and $c_i(f)$ so that the total fractional cost $c_i(e)\cdot x_e+c_i(f)\cdot x_f$ is preserved exactly. We always move to an endpoint at which at least one of the two coordinates becomes integral, and randomize between the two possible directions so that each coordinate preserves its expectation. Repeating this operation leaves at most one fractional coordinate. Thus, after the pair steps, all remaining rounding randomness is concentrated in a single chore: the final step chooses between two integral allocations that differ only in that chore. We call this procedure \wpralg.

\begin{algorithm}[htp]
\caption{\wpralg}
\label{alg:wpr}
\KwIn{A set $R$ of items, the current agent's positive costs $(c_i(e))_{e\in R}$, and $p\in[0,1]$.}
\KwOut{A random bundle $X\subseteq R$.}
Fix an ordering of $R$ and set $x_e\leftarrow p$ for every $e\in R$\;
\While{at least two coordinates lie strictly between zero and one}{
  Let $e,f$ be the first two fractional coordinates\;
  Set $\ell\leftarrow\max\{-c_i(e)\cdot x_e,-c_i(f)\cdot (1-x_f)\}$ and $u\leftarrow\min\{c_i(e)\cdot (1-x_e),c_i(f)\cdot x_f\}$\;
  Draw $s=u$ with probability $-\ell/(u-\ell)$, and $s=\ell$ with probability $u/(u-\ell)$\;
  Set $x_e\leftarrow x_e+s/c_i(e)$ and $x_f\leftarrow x_f-s/c_i(f)$\;
}
\If{one fractional coordinate $x_e$ remains}{
  Set $x_e\leftarrow1$ with probability $x_e$, and $x_e\leftarrow0$ otherwise\;
}
\Return{$X=\{e\in R:x_e=1\}$}\;
\end{algorithm}

\begin{lemma}
\label{lem:pair-step}
Every pair step in {\normalfont\wpralg} keeps all coordinates in $[0,1]$, preserves the conditional expectation of each coordinate, preserves the total fractional cost exactly, and makes at least one of the two updated coordinates integral. Consequently, after at most $|R|-1$ pair steps, at most one fractional coordinate remains.
\end{lemma}

The proof is a direct calculation from the definitions of $\ell$ and $u$ and is deferred to Appendix~\ref{app:additive-expectation-proofs}.

The next two lemmas record the two properties of the rounding procedure that will be used later. The first controls the current agent's cost.

\begin{lemma}
\label{lem:rounding-cost}
Let $X$ be the output of {\normalfont\wpralg} on the item set $R$ with common initial fraction $p$. There exist $x\in[0,1]$, $0\le d\le\max_{e\in R}\{c_i(e)\}$, and $Z\sim\operatorname{Bernoulli}(x)$ such that
\begin{equation*}
  c_i(X)=p\cdot c_i(R)+d\cdot (Z-x).
\end{equation*}
Consequently, we have
\begin{equation*}
  \E[c_i(X)]=p\cdot c_i(R),
  \quad \text{ and } \quad
  c_i(X)\le p\cdot c_i(R)+\max_{e\in R}\{c_i(e)\}.
\end{equation*}
\end{lemma}
\begin{proof}
After the pair steps, at most one coordinate is fractional and the total fractional cost is still $p\cdot c_i(R)$. If no fractional coordinate remains, take $d=x=Z=0$. Otherwise, let $x\in(0,1)$ be the remaining fractional coordinate and let $d$ be the cost of its chore. The integral coordinates then have total cost $p\cdot c_i(R)-d\cdot x$. The final step replaces $x$ by $Z\sim\operatorname{Bernoulli}(x)$, which gives the first statement. Taking expectation gives the first consequence, while $Z-x\le 1$ gives the second.
\end{proof}

\begin{lemma}
\label{lem:rounding-product}
Let $X$ be the output of {\normalfont\wpralg} on the item set $R$ with common fraction $p$. For every deterministic vector $(z_e)_{e\in R}$ with $z_e\ge1$,
\begin{equation*}
  \E\left[\prod_{e\in R\setminus X}z_e\right]
  \le
  \prod_{e\in R}\left(p+(1-p)z_e\right).
\end{equation*}
\end{lemma}

\begin{proof}
Define $G(\boldsymbol{x}) :=\prod_{e\in R}\left(1+(1-x_e)(z_e-1)\right)$.
Consider one pair step on $e,f$. Along this step, all factors except those for $e$ and $f$ are fixed, while these two factors are affine in $s$ with opposite slopes. Hence
\begin{equation*}
  \frac{d^2}{ds^2}G(\boldsymbol{x}(s))
  =-\frac{2(z_e-1)(z_f-1)}{c_i(e)c_i(f)}
    \prod_{g\in R\setminus\{e,f\}}
    \left(1+(1-x_g)(z_g-1)\right)
  \le0.
\end{equation*}
Thus $G(\boldsymbol{x}(s))$ is concave in $s$. Since $\E[s\mid\boldsymbol{x}]=0$, Jensen's inequality gives $\E\left[G(\boldsymbol{x}^{\mathrm{after}})\mid\boldsymbol{x}^{\mathrm{before}}\right]\le G(\boldsymbol{x}^{\mathrm{before}})$.
The final Bernoulli step preserves the conditional expectation of $G$: if the remaining fractional coordinate is $x_e$, its factor is $1+(1-x_e)(z_e-1)$ before rounding, while after rounding its expected factor is $x_e+(1-x_e)z_e$, which is the same quantity. Therefore
\begin{equation*}
  \E\left[G(\boldsymbol{x}^{\mathrm{final}})\right]
  \le G(\boldsymbol{x}^{\mathrm{initial}})
  =\prod_{e\in R}\left(p+(1-p)z_e\right).
\end{equation*}
At the final integral vector, $G(\boldsymbol{x}^{\mathrm{final}})=\prod_{e\in R\setminus X}z_e$, which proves the lemma.
\end{proof}

In round $i$, we apply \wpralg\ to $R_i$ with common initial fraction $p_i$. Algorithm~\ref{alg:scheduled-weighted} applies this rounding procedure in every round.

\begin{algorithm}[htp]
\caption{\scheduledalg}
\label{alg:scheduled-weighted}
\KwIn{$n$, $M$, and a schedule $(p_1,\ldots,p_n)$ chosen before the agents arrive.}
\KwOut{A random complete allocation $X=(X_1,\ldots,X_n)$.}
Set $R\leftarrow M$\;
\For{$i=1,\ldots,n$}{
  Observe $c_i$\;
  $X_i\leftarrow\wpralg(R,(c_i(e))_{e\in R},p_i)$\;
  Set $R\leftarrow R\setminus X_i$\;
}
\Return{$X$}\;
\end{algorithm}

All item costs passed to the rounding subroutine are positive by the perturbation described at the beginning of this section. Allocation decisions are irrevocable and use only the current cost function and unallocated chores. Since $p_n=1$, the last agent receives every chore still unallocated. Thus Algorithm~\ref{alg:scheduled-weighted} is a complete online allocation rule. The normalization above is used only in the analysis and does not require the algorithm to compute any MMS value.

\subsection{Load Analysis for a General Schedule}
\label{sec:additive-load-analysis}
\label{sec:additive-expectation}

We now analyze Algorithm~\ref{alg:scheduled-weighted} for an arbitrary schedule $(p_1,\ldots,p_n)$. The goal of this subsection is to obtain a probability bound for the fractional cost of each arriving agent. We will choose the schedule only in the next subsection.

Define $s_1:=1$ and, for every $i\in[n]$, $s_{i+1}:=s_i\cdot(1-p_i)$.
Thus $s_i=\prod_{j=1}^{i-1}(1-p_j)$. The quantity $s_i$ has a simple benchmark interpretation: if every chore were independently allocated in round $j$ with probability $p_j$, then $s_i$ would be the probability that the chore survives the first $i-1$ rounds and remains available before round $i$. We do not claim that the actual remaining chores are independent. Instead, $s_i$ describes the independent-retention process with which we compare the actual dependent process. Since $p_n=1$, we have $s_{n+1}=0$.

Applying Lemma~\ref{lem:rounding-cost} with the item set $R=R_i$ and $p=p_i$, and using $c_i(e)\le1$, gives
\begin{equation*}
  c_i(X_i)\le p_i\cdot c_i(R_i)+1,
  \qquad
  \E[c_i(X_i)\mid R_i]=p_i\cdot c_i(R_i).
\end{equation*}
For the last round, $p_n=1$, so $c_n(X_n)=p_n\cdot c_n(R_n)$. Consequently, every realization satisfies
$\rho(X)\le1+\max_{i\in[n]}\left\{p_i\cdot c_i(R_i)\right\}$.
Thus both the expectation and high-probability analyses reduce to controlling the random fractional costs $p_i\cdot c_i(R_i)$.

This is the same probabilistic difficulty that appeared in the binary section: selections in different rounds are dependent because selected chores disappear. There we compared the actual process with independent sampling from the full cores. Here the chores have unequal weights, so cardinality is no longer the right state variable. Instead, the next invariant compares the random remaining item set with independent retention at the level of products. Substituting exponential test values into this invariant will then give the same Chernoff rate function $\psi$ as in the binary analysis.

\begin{lemma}\label{lem:remaining-invariant}
For every $i\in[n]$ and every deterministic vector $(z_e)_{e\in M}$ with $z_e\ge1$,
\begin{equation}
  \E\left[\prod_{e\in R_i}z_e\right]
  \le\prod_{e\in M}\left(1-s_i+s_i z_e\right).
  \label{eq:remaining-set}
\end{equation}
\end{lemma}

The right-hand side is exactly the product moment obtained if every chore were retained independently with probability $s_i$. The proof is given in Appendix~\ref{app:additive-expectation-proofs}.

The lemma also justifies the benchmark interpretation of $s_i$ at the level of individual chores. Setting one coordinate $z_e$ equal to two and all others equal to one in \eqref{eq:remaining-set} gives
\begin{equation*}
  1+\Prb[e\in R_i]\le1+s_i,
  \qquad\text{hence}\qquad
  \Prb[e\in R_i]\le s_i.
\end{equation*}
For agent $i$, linearity of expectation and $c_i(M)\le n$ therefore imply
\begin{equation*}
  \E\left[p_i\cdot c_i(R_i)\right]
  =p_i\cdot \sum_{e\in M}c_i(e)\cdot \Prb[e\in R_i]
  \le p_i\cdot s_i\cdot c_i(M)
  \le n\cdot s_i\cdot p_i.
\end{equation*}
These marginal bounds alone do not control the maximum over all agents. We therefore use the full product invariant to obtain an exponential-moment bound and a corresponding probability bound. For agent $i$ and $\lambda>0$, choosing $z_e=\exp(\lambda\cdot p_i\cdot c_i(e))$ for each chore $e$ makes the left-hand side of the invariant equal to $\E\left[e^{\lambda p_i c_i(R_i)}\right]$.
Recall the Chernoff rate function $\psi$ from the binary analysis. It satisfies $\psi(1)=0$, $\psi'(x)=\log x$, and $\psi''(x)=1/x>0$.

\begin{lemma}\label{lem:scheduled-load}
For every round $i$ and every $\lambda>0$,
\begin{equation}
  \E\left[e^{\lambda p_i c_i(R_i)}\right]
  \le\left(1-s_i+s_i e^{\lambda p_i}\right)^n.
  \label{eq:fractional-load-mgf}
\end{equation}
Consequently, for every $u\ge n\cdot s_i\cdot p_i$,
\begin{equation}
  \Prb\left[p_i\cdot c_i(R_i)\ge u\right]
  \le\exp\left[-n\cdot s_i\cdot\psi\left(\frac{u}{n\cdot s_i\cdot p_i}\right)\right].
  \label{eq:scheduled-chernoff}
\end{equation}
\end{lemma}

The proof is given in Appendix~\ref{app:additive-expectation-proofs}. Lemma~\ref{lem:scheduled-load} is the probability estimate for a general schedule that we will use below when choosing the schedule.

\subsection{Choosing the Schedule}
\label{sec:additive-schedule-choice}

The preceding analysis applies to every schedule $(p_1,\ldots,p_n)$. We now choose these fractions so that the probability bounds are simultaneously small for all agents. It is convenient to parameterize the schedule by $\beta_i:=n\cdot s_i\cdot p_i$ for all $i\in[n]$.
The quantity $\beta_i$ is the upper bound on the expected fractional cost obtained above. Define $B_i:=n\cdot s_i$. Since $s_{i+1}=s_i\cdot(1-p_i)$,
\begin{equation*}
  B_{i+1}=B_i-\beta_i,
  \qquad
  B_i=\sum_{j=i}^n\beta_j,
  \qquad
  p_i=\frac{\beta_i}{B_i}.
\end{equation*}
Moreover, $B_1=n$ and $B_{n+1}=0$, so $\sum_{i=1}^n\beta_i=n$. Conversely, any positive vector $(\beta_1,\ldots,\beta_n)$ summing to $n$ determines a schedule through the preceding formulas. This parameterization is useful only for choosing the fractions $p_i$; the allocation rule itself continues to use only the schedule $(p_i)$.
We fix $K:=2\log n$. Recall that $\psi(x)=x\log x-x+1$. Before any agent arrives, compute auxiliary values $\hat\beta_i$ backward. Set $\hat B_{n+1}:=0$. For $i=n,n-1,\ldots,1$, choose the unique $\hat\beta_i\in(0,1)$ satisfying
\begin{equation*}
  (\hat B_{i+1}+\hat\beta_i)\cdot\psi(1/\hat\beta_i)=K,
  \qquad
  \hat B_i:=\hat B_{i+1}+\hat\beta_i.
\end{equation*}
Such a value exists because, for every fixed $C\ge0$, the function $(C+x)\cdot\psi(1/x)$ is continuous and strictly decreasing from $+\infty$ to zero on $(0,1)$.
Set
\begin{equation*}
  \sigma:=\frac{n}{\hat B_1},
  \qquad
  \beta_i:=\sigma\cdot\hat\beta_i,
  \qquad
  B_i:=\sigma\cdot\hat B_i,
  \qquad
  p_i:=\frac{\beta_i}{B_i}.
\end{equation*}

We will use two properties of this construction. The first bounds the normalization factor $\sigma$. Its proof is the only substantial calculation concerning the backward recurrence and is deferred to the appendix.

\begin{lemma}\label{lem:schedule-deficit}
For the fixed choice $K=2\log n$,
\begin{equation*}
  n-\frac n\sigma
  =\sum_{i=1}^n\left(1-\frac{\beta_i}{\sigma}\right)
  \le(4+o(1))\sqrt{n\log n}.
\end{equation*}
Consequently,
$\sigma-1\le(4+o(1))\sqrt{\frac{\log n}{n}}$.
\end{lemma}

The proof is given in Appendix~\ref{app:additive-expectation-proofs}.
The second property explains the choice of the recurrence: it makes the Chernoff exponent at the common threshold $\sigma$ identical in every round.

\begin{lemma}\label{lem:schedule-construction}
For the above backward construction, $(p_1,\ldots,p_n)$ is a schedule. Moreover, we have $\sigma>1$, $0<\beta_i<\sigma$ for all $i\in[n]$, and
\begin{equation*}
  B_i\cdot \psi\left(\frac{\sigma}{\beta_i}\right)=\sigma\cdot K
  =2\sigma\cdot \log n
  \qquad \forall i\in[n].
\end{equation*}
\end{lemma}
\begin{proof}
Since every $\hat\beta_i<1$, we have $0<\hat B_1=\sum_{i=1}^n\hat\beta_i<n$, and hence $\sigma=n/\hat B_1>1$. The normalized values satisfy $\sum_{i=1}^n\beta_i=n$, so they determine a schedule. Also, $0<\beta_i<\sigma$. Finally, $\hat B_i=\hat B_{i+1}+\hat\beta_i$, so multiplying the defining recurrence by $\sigma$ and using $B_i=\sigma\cdot \hat B_i$ and $\beta_i=\sigma\cdot \hat\beta_i$ gives the last statement.
\end{proof}

We can now prove the high-probability part of Theorem~\ref{thm:additive-high-probability}. By the deterministic bound from Section~\ref{sec:additive-load-analysis}, if $\rho(X)\ge1+\sigma$, then $p_i\cdot c_i(R_i)\ge\sigma$ for some agent $i$. Because $n\cdot s_i\cdot p_i=\beta_i<\sigma$, a union bound and Lemma~\ref{lem:scheduled-load} give
\begin{align*}
  \Prb[\rho(X)\ge1+\sigma]
  \le\sum_{i=1}^n
  \Prb\left[p_i\cdot c_i(R_i)\ge\sigma\right]
  \le\sum_{i=1}^n
  \exp\left[-B_i\cdot\psi\left(\frac{\sigma}{\beta_i}\right)\right]
  =n e^{-\sigma K}
  \le\frac1n,
\end{align*}
where the last inequality uses $\sigma>1$ and $K=2\log n$. Therefore, with probability at least $1-1/n$,
\begin{equation*}
  \rho(X)<1+\sigma
  =2+(\sigma-1)
  \le2+(4+o(1))\sqrt{\frac{\log n}{n}},
\end{equation*}
where the last inequality follows from Lemma~\ref{lem:schedule-deficit}.

\smallskip

For the expectation bound, we use the following elementary integral estimate.

\begin{lemma}
\label{lem:convex-integral}
Let $f$ be differentiable and convex on $[a,\infty)$, and suppose that $f'(a)>0$. Then
\begin{equation*}
  \int_a^\infty e^{-f(u)}\,du
  \le
  \frac{e^{-f(a)}}{f'(a)}.
\end{equation*}
\end{lemma}
\begin{proof}
Convexity gives $f(u)\ge f(a)+f'(a)(u-a)$ for every $u\ge a$. Since $f'(a)>0$,
\begin{equation*}
  \int_a^\infty e^{-f(u)}du
  \le e^{-f(a)}\cdot \int_0^\infty e^{-f'(a)v}dv
  =\frac{e^{-f(a)}}{f'(a)}.
  \qedhere
\end{equation*}
\end{proof}

We now prove the expectation bound. Again using
$\rho(X)\le1+\max_{i\in[n]}\{p_i\cdot c_i(R_i)\}$, the integral representation of the nonnegative excess above $\sigma$, a union bound, and Lemma~\ref{lem:scheduled-load} give
\begin{equation*}
  \E[\rho(X)]
  \le1+\sigma+
  \sum_{i=1}^n
  \int_\sigma^\infty
  e^{-B_i\psi(u/\beta_i)}\,du.
\end{equation*}
Fix a round $i$ and let $f(u):=B_i\cdot \psi(u/\beta_i)$. We have
\begin{equation*}
  f'(u)=\frac{B_i}{\beta_i}\cdot \log\left(\frac{u}{\beta_i}\right),
  \qquad
  f''(u)=\frac{B_i}{\beta_i}\cdot\frac1u>0.
\end{equation*}
Thus $f$ is convex. Moreover, $f(\beta_i)=0$ and Lemma~\ref{lem:schedule-construction} gives $f(\sigma)=\sigma K$. Since $\beta_i<\sigma$, convexity implies $f'(\sigma)\ge\frac{f(\sigma)-f(\beta_i)}{\sigma-\beta_i} = \frac{\sigma K}{\sigma-\beta_i}>0$.
Lemma~\ref{lem:convex-integral} therefore gives
\begin{equation*}
  \int_\sigma^\infty e^{-B_i\psi(u/\beta_i)}\,du
  \le e^{-\sigma K}\cdot\frac{\sigma-\beta_i}{\sigma K}.
\end{equation*}
Summing over the rounds and using $\sum_{i=1}^n\beta_i=n$ gives
\begin{equation*}
  \E[\rho(X)]
  \le1+\sigma+
  \frac{n(\sigma-1)}{\sigma K}e^{-\sigma K}.
\end{equation*}
Since $\sigma>1$ and $K=2\log n$, the last term satisfies
\begin{equation*}
  \frac{n(\sigma-1)}{\sigma K}e^{-\sigma K}
  \le \frac{\sigma-1}{2\sigma n\log n}
  =o\left(\sqrt{\frac{\log n}{n}}\right).
\end{equation*}
Combining this with Lemma~\ref{lem:schedule-deficit} yields $\E[\rho(X)] \le2+(4+o(1))\sqrt{\frac{\log n}{n}}$.
This proves the expectation statement in Theorem~\ref{thm:additive-high-probability}.

The preceding schedule is designed for the asymptotic guarantee and may be loose for small $n$. By setting the schedule differently, we obtain a uniform expected-ratio guarantee below $2.594$.

\begin{theorem}\label{thm:additive-uniform-bound}
For every $n\ge2$, there exists a schedule for Algorithm~\ref{alg:scheduled-weighted} such that the returned complete allocation $X$ satisfies
$\E[\rho(X)]<2.594$.
\end{theorem}

The explicit schedule and its proof are given in Appendix~\ref{app:schedule_for_uniform}.




\section*{Declaration for the Use of AI}
We used OpenAI's GPT-6 Astra to assist with mathematical analysis, language editing, notation consistency checks, and LaTeX formatting. All mathematical statements, proofs, calculations, and citations were reviewed and verified by the authors, who take full responsibility for the content of the paper.

\newpage
\bibliographystyle{alpha}
\bibliography{Online_MMS}

\newpage

\appendix

\section{Missing Proofs from the Binary-Cost Analysis}
\label{app:binary}

\begin{proofof}{Lemma~\ref{lem:binary-convex}}
Append $C_{\ell+1}:=Q$ to the nested sequence of cores and put $p_i:=\frac{|C_{i+1}|}{|C_i|}$ for all $i\in[\ell]$.
We analyze the selections by marking every chore that has already been selected. Although a marked chore is no longer available to the algorithm, we keep its location when we restrict a core to the next nested core. Conditional on the number of marks in a core, their locations are uniform in that core. This property holds initially and is preserved both when uniformly chosen unmarked chores are marked and when the core is restricted to a fixed subset.

The only complication is that the algorithm may exhaust the currently unallocated part of a core. This is exactly why the indicator appears in the statement of the lemma. If a selection exhausts $C_i$, every still-unallocated chore of $Q\subseteq C_i$ is selected in that round. Hence $H_Q=|Q|$ at the end, and the random variable
$\phi(H_Q)\cdot\mathbf 1_{\{H_Q<|Q|\}}$
vanishes on this branch. We may therefore upper bound exhausted branches by zero and analyze only nonexhausted selections.

We first record the property that allows the same comparison to be repeated through all nested cores. Let $f$ be nonnegative, nonincreasing, and convex on the nonnegative integers, let $0\le p\le1$, and define $g(s):=\E[f(\operatorname{Bin}(s,p))]$ for $s=0,1,\ldots$.
Couple $\operatorname{Bin}(s+1,p)$ with $\operatorname{Bin}(s,p)$ by adding one independent Bernoulli variable. If $B_s\sim\operatorname{Bin}(s,p)$, then
\begin{align*}
  g(s+1)-g(s)
  &=p\cdot\E[f(B_s+1)-f(B_s)]\le0,\\
  g(s+2)-2g(s+1)+g(s)
  &=p^2\cdot\E[f(B_s+2)-2f(B_s+1)+f(B_s)]\ge0.
\end{align*}
Thus binomial averaging preserves nonnegativity, monotonicity, and convexity. Replacing $s$ by $s+q$ for a fixed nonnegative integer $q$ preserves the same three properties.

We now compare the actual process with a relaxed process by backward induction. At the terminal core $Q=C_{\ell+1}$, use the relaxed continuation value
$f_{\ell+1}(s):=\phi(s)$.
The actual terminal payoff is $\phi(s)\cdot\mathbf 1_{\{s<|Q|\}} \le f_{\ell+1}(s)$, so the desired domination holds at the terminal stage.

Suppose the relaxed continuation function $f_{i+1}$ has already been defined and is nonnegative, nonincreasing, and convex. Define
\begin{equation*}
  f_i(s):=
  \E\left[f_{i+1}\left(\operatorname{Bin}(s+q_i,p_i)\right)\right].
\end{equation*}
By the preceding preservation property, $f_i$ again has all three properties.

Consider an actual state with $s$ marked locations in $C_i$. If the round-$i$ selection exhausts the unallocated part of $C_i$, its eventual truncated payoff is zero and is therefore at most the nonnegative quantity $f_i(s)$. Otherwise, the round marks exactly $q_i$ new locations. Before the restriction to $C_{i+1}$ there are therefore $s+q_i$ marked locations, uniformly distributed in $C_i$. The number $S'$ of these marks lying in $C_{i+1}$ has a hypergeometric distribution. Hoeffding's convex comparison for sampling without replacement~\citep{hoeffding1963probability} gives
\begin{equation*}
  \E[f_{i+1}(S')]
  \le
  \E\left[f_{i+1}\left(\operatorname{Bin}(s+q_i,p_i)\right)\right]
  =f_i(s).
\end{equation*}
Combining this inequality with the induction hypothesis for the continuation after the restriction proves that the actual expected truncated payoff from state $s$ is at most $f_i(s)$. Induction from $i=\ell$ down to $i=1$ therefore yields $\E\left[\phi(H_Q)\cdot\mathbf 1_{\{H_Q<|Q|\}}\right]
  \le f_1(0)$.

It remains to identify the relaxed terminal count represented by $f_1(0)$. In the relaxed process, each of the $q_i$ marks introduced in round $i$ survives the successive independent binomial restrictions with probability $p_i\cdot p_{i+1}\cdots p_\ell =\frac{|Q|}{|C_i|}$.
All these survival events are independent in the relaxed process. Hence the terminal count is distributed as $\sum_{i=1}^{\ell} \operatorname{Bin}\left(q_i,\frac{|Q|}{|C_i|}\right)$, with independent summands, which is exactly the auxiliary variable $K$. Thus $\E\left[\phi(H_Q)\cdot\mathbf 1_{\{H_Q<|Q|\}}\right] \le \E[\phi(K)]$.
 
For the second inequality, $K$ is a sum of independent Bernoulli random variables with mean $\E[K]$. A standard Poisson convex-order comparison states that such a sum is dominated in convex order by a Poisson random variable with the same mean; see, e.g., \citep[Theorem~2.2]{daly2016negative}. Since $\phi$ is convex, $\E[\phi(K)] \le \E[\phi(\operatorname{Pois}(\E[K]))]$.
Finally, if $\phi(|Q|)=0$, then $\phi(H_Q)\cdot\mathbf 1_{\{H_Q<|Q|\}}=\phi(H_Q)$, which proves the last statement of the lemma.
\end{proofof}

\begin{proofof}{Corollary~\ref{cor:binary-selection-tail}}
If $y=0$, the event $H_Q<y$ is empty. If $y=\lambda$, both stated
upper bounds are one. We therefore assume $0<y<\lambda$.
Fix $\theta>0$ and apply Lemma~\ref{lem:binary-convex} with
$\phi(w)=e^{-\theta w}$. Since this function is nonnegative,
nonincreasing, and convex,
\begin{equation*}
  \E\left[e^{-\theta H_Q}\cdot\mathbf 1_{\{H_Q<|Q|\}}\right]
  \le
  \E[e^{-\theta K}]
  \le
  \exp\left(\E[K]\cdot\left(e^{-\theta}-1\right)\right).
\end{equation*}
The second inequality follows from the Poisson comparison in
Lemma~\ref{lem:binary-convex}; equivalently, it is the standard
exponential-moment bound for a sum of independent Bernoulli variables.

On the event $H_Q<y$, we have $H_Q<|Q|$ and
$e^{-\theta H_Q}>e^{-\theta y}$. Markov's inequality therefore gives
\begin{equation*}
  \Prb[H_Q<y]
  \le
  \exp\left(\theta y+\E[K]\cdot\left(e^{-\theta}-1\right)\right)
  \le
  \exp\left(\theta y+\lambda\cdot\left(e^{-\theta}-1\right)\right),
\end{equation*}
where the last step uses $\E[K]\ge\lambda$ and
$e^{-\theta}-1<0$. Choosing $\theta=\log(\lambda/y)$ makes
$e^{-\theta}=y/\lambda$ and gives
\begin{equation*}
  \Prb[H_Q<y]
  \le
  \exp\left[-\lambda\psi\left(\frac{y}{\lambda}\right)\right].
\end{equation*}
Finally, $0<y/\lambda<1$, so the quadratic lower bound on $\psi$ gives
$\lambda\psi(y/\lambda)\ge(\lambda-y)^2/(2\lambda)$, proving the
second inequality.
\end{proofof}

\begin{proofof}{Lemma~\ref{lem:binary-failure-two}}
Fix the transition index $h$. If $h=0$, the core is small from the first round and the algorithm assigns at most one costly chore to each agent, so $\rho(X)\le1$. We may therefore assume that $h\ge1$.
First suppose that $h<n$. The first $h$ agents have ratio at most $3/2$. By Lemma~\ref{lem:binary-deterministic}, every later agent has ratio at most $\max\left\{1,\left\lceil\frac{n-H}{n-h}\right\rceil\right\}$.
If $h\le n/2$, then $(n-H)/(n-h)\le n/(n-h)\le2$, so a ratio above two is impossible. We may therefore assume that $h>n/2$.
A ratio above two requires
$H<2h-n<h$.
Condition on $T=t$. Since $t\le L<n/2<h$, Lemma~\ref{lem:binary-insurance-tail} gives
\begin{equation*}
  \Prb[\rho(X)>2\mid T=t]
  \le
  \Prb[H<h\mid T=t]
  \le
  \exp\left(-\frac{(h-t)^2}{24h}\right)
  \le
  \exp\left(-\frac{(h-L)^2}{24h}\right).
\end{equation*}
For $h>L$, the function $(h-L)^2/h=h-2L+L^2/h$ is increasing in $h$. Since $h>n/2$, we obtain
\begin{equation*}
  \frac{(h-L)^2}{24h}
  \ge
  \frac{(n/2-L)^2}{24(n/2)}
  =\frac{(n-2L)^2}{48n}.
\end{equation*}
Averaging over $T$ proves the claimed bound whenever $h<n$.

It remains to consider $h=n$. Let $H_{C_n}$ be the number of chores of $C_n$ selected by the prescribed quotas in the $n$ rounds before the final cleanup. Conditional on $T=t\le L$, apply Corollary~\ref{cor:binary-selection-tail} with $\ell=n$ and $Q=C_n$. Every aggressive round contributes at least $|C_n|/n$ to the auxiliary mean, while every cautious round contributes at least $4|C_n|/(3n)$. Hence
\begin{equation*}
  \E[K]
  \ge
  |C_n|\left(1+\frac{n-t}{3n}\right)
  \ge
  |C_n|\left(1+\frac{n-L}{3n}\right).
\end{equation*}
Set
\begin{equation*}
  a:=\frac{n-L}{3n},
  \qquad
  \lambda:=(1+a)|C_n|,
  \qquad
  y:=|C_n|.
\end{equation*}
Since $0<a\le1/3$, Corollary~\ref{cor:binary-selection-tail} gives
\begin{equation*}
  \Prb[H_{C_n}<|C_n|\mid T=t]
  \le \exp\left(-\frac{|C_n|a^2}{2(1+a)}\right)
  \le \exp\left(-\frac{(n-L)^2}{24n}\right).
\end{equation*}
Here we used $|C_n|>n$ and $1+a\le4/3$. If $H_{C_n}=|C_n|$, every chore of $C_n$ has already been selected before the final cleanup. The cleanup is then empty, and every agent has ratio at most $3/2$. Therefore the preceding display also bounds $\Prb[\rho(X)>2\mid T=t]$.

Finally,
$2(n-L)^2-(n-2L)^2=n^2-2L^2>0$
under $L<n/2$. Hence
$\frac{(n-L)^2}{24n} \ge \frac{(n-2L)^2}{48n}$.
This proves the claimed bound also when $h=n$.
\end{proofof}

\section{Binary Parameter Choices}
\label{app:binary-parameters}

This appendix gives the direct calculations used for the three small-agent guarantees in Theorem~\ref{thm:binary-all-agreesive}, together with the short numerical verification needed for the uniform bound in Theorem~\ref{thm:binary-immediate}. No finite-$n$ refinement beyond $n=4$ is needed for the results stated in the main text.

\paragraph{Two agents.}
Set $T=2$. The first agent receives every zero-cost chore and a uniformly random $\lceil|B_1|/2\rceil$-subset of $B_1$; the second agent receives every chore that remains. The first agent's ratio is at most one.

Let $C:=B_1\cap B_2$ be the common costly-item set. Since $|C|\le |B_2|\le 2\mu_2$, partition $C$ into exactly $\mu_2$ blocks, each containing at most two chores; empty blocks are allowed. For a block $b$, let $Z_b$ be the number of its chores left for agent~2. Then
\begin{equation*}
  \rho(X)
  \le 1+\frac1{\mu_2}\sum_{b=1}^{\mu_2}(Z_b-1)^+.
\end{equation*}
A block of size zero or one contributes zero. For a two-chore block, $(Z_b-1)^+$ is the indicator that both chores remain after the first allocation. If $m:=|B_1|$ and $r:=\lfloor m/2\rfloor$, this probability is
$\frac{r(r-1)}{m(m-1)}\le\frac14$. Therefore $\E[\rho(X)]\le 1+1/4=5/4$, matching the lower bound.

\paragraph{Three agents.}
Set $T=3$ and put $k:=|C_2|$. For the upper-bound calculation, retain only the base quota from each of the first two selected bundles. Partition $C_3$ into blocks of size at most three. If $Z$ is the number of chores left in one block, it suffices to bound $\E[(Z-1)^+]$. Conditional on a remaining subset of $C_2$ of size $\ell$,
\begin{equation*}
  \E[(Z-1)^+]
  =\frac{3\ell(\ell-1)}{k(k-1)}
   -\frac{\ell(\ell-1)(\ell-2)}{k(k-1)(k-2)}.
\end{equation*}
After substituting the number left by agent~2, the right-hand side is increasing and convex in the number left by agent~1. Hoeffding's comparison therefore reduces the calculation to $Y\sim\operatorname{Bin}(k,2/3)$. Evaluating the resulting three residue classes of $k$ modulo three shows that the maximum excess is $1088/3645$, attained at $k=6$. Hence
$\E[\rho(X)]\le 1+1088/3645=4733/3645<1.2985$.

\paragraph{Four agents.}
Set $T=4$. If the core becomes small after $h\le3$ rounds, Lemma~\ref{lem:binary-convex} gives
\begin{equation*}
  \E[\rho(X)]
  \le1+\max_{1\le h\le3}\E[(h-\operatorname{Pois}(h))^+]
  =1+\frac{27}{2e^3}<1.6722.
\end{equation*}
If all four rounds have large cores, partition $C_4$ into blocks of size at most four and apply the same estimate to each block.

\paragraph{Uniform always-cautious bound.}
In the proof of Theorem~\ref{thm:binary-immediate}, the only numerical quantity that remains is
$s_k:=\E[(k-\frac12-\operatorname{Pois}(4k/3))^+]$. Direct evaluation of the finite Poisson sums gives
\begin{equation*}
  \max_{k\ge0}s_k=s_6=\frac{22781}{30e^8}<0.25474.
\end{equation*}
For completeness, the infinite tail can be bounded without numerical search. The elementary inequality $y^+\le3e^{y/3-1}$ gives $s_k\le3e^{-7/6-k/25}$, which is already below $s_6$ for all sufficiently large $k$; the remaining finitely many values are obtained by direct substitution in the Poisson formula. Consequently, Theorem~\ref{thm:binary-immediate} yields $\E[\rho(X)]<1.755$ for every $n$.

\section{Missing Proofs from the General-Additive Analysis}
\label{app:additive-expectation-proofs}

This appendix contains the routine proofs deferred from the general-additive analysis.

\begin{proofof}{Lemma~\ref{lem:pair-step}}
Consider a pair step on two fractional coordinates $x_e,x_f\in(0,1)$. If the displacement is $s$, then
\begin{equation*}
  x_e(s)=x_e+\frac{s}{c_i(e)},
  \qquad
  x_f(s)=x_f-\frac{s}{c_i(f)}.
\end{equation*}
The constraints $0\le x_e(s),x_f(s)\le1$ are equivalent to
\begin{equation*}
  -c_i(e)x_e\le s\le c_i(e)(1-x_e),
  \qquad
  -c_i(f)(1-x_f)\le s\le c_i(f)x_f.
\end{equation*}
Their intersection is exactly $[\ell,u]$. Since $x_e,x_f\in(0,1)$ and the item costs are positive, $\ell<0<u$, and at either endpoint at least one coordinate becomes integral. Thus every update is feasible and reduces the number of fractional coordinates by at least one.

By the choice of the endpoint probabilities,
\begin{equation*}
  \E[s\mid\boldsymbol{x}]
  =\frac{-\ell}{u-\ell}\cdot u
   +\frac{u}{u-\ell}\cdot\ell
  =0.
\end{equation*}
Hence both updated coordinates preserve their conditional expectations. Moreover,
\begin{equation*}
  c_i(e)x_e(s)+c_i(f)x_f(s)
  =c_i(e)x_e+c_i(f)x_f,
\end{equation*}
so the total fractional cost is preserved in every outcome. Since each pair step makes at least one fractional coordinate integral and integral coordinates are never selected again, after at most $|R|-1$ pair steps at most one fractional coordinate remains.
\end{proofof}

\begin{proofof}{Lemma~\ref{lem:remaining-invariant}}
We induct upward on $i$. Before round $1$, we have $R_1=M$ and $s_1=1$, so both sides of \eqref{eq:remaining-set} equal $\prod_{e\in M}z_e$.

Suppose that \eqref{eq:remaining-set} holds before round $i<n$, and fix a deterministic test vector $(z_e)_{e\in M}$. Conditional on $R_i$, agent $i$ applies weighted pairwise rounding to the item set $R_i$ and receives the bundle $X_i\subseteq R_i$. Thus $R_{i+1}=R_i\setminus X_i$. By Lemma~\ref{lem:rounding-product},
\begin{equation*}
  \E\left[\left.\prod_{e\in R_{i+1}}z_e\,\right|R_i\right]
  \le\prod_{e\in R_i}\left(p_i+(1-p_i)z_e\right).
\end{equation*}

The vector with coordinates $p_i+(1-p_i)z_e$ is deterministic and has every coordinate at least one. Taking expectations and applying the induction hypothesis to this vector yields
\begin{align*}
  \E\left[\prod_{e\in R_{i+1}}z_e\right]
  \le\prod_{e\in M}
     \left(1-s_i+s_i\left(p_i+(1-p_i)z_e\right)\right)
  =\prod_{e\in M}
     \left(1-s_{i+1}+s_{i+1}z_e\right),
\end{align*}
where the last equality uses $s_{i+1}=s_i\cdot(1-p_i)$. This is the invariant before round $i+1$.
\end{proofof}

\begin{proofof}{Lemma~\ref{lem:scheduled-load}}
Fix $i$. Under the normalization in Section~4, $0\le c_i(e)\le1$ for every chore and $c_i(M)\le n$. Because the adversary is oblivious, these costs are fixed independently of the algorithm's random choices, even before agent $i$ arrives. Thus $z_e=e^{\lambda p_i c_i(e)}$ is an admissible deterministic test vector in \eqref{eq:remaining-set}. Substitution gives
\begin{equation*}
  \E\left[e^{\lambda p_i c_i(R_i)}\right]
  =\E\left[\prod_{e\in R_i}e^{\lambda p_i c_i(e)}\right]
  \le\prod_{e\in M}
     \left(1-s_i+s_i e^{\lambda p_i c_i(e)}\right).
\end{equation*}

To bound this product, let $g(x):=\log\left(1-s_i+s_i e^{\lambda p_i x}\right)$. Then $g(0)=0$, $g(1)\ge0$, and
\begin{equation*}
  g''(x)=
  \frac{s_i(1-s_i)(\lambda p_i)^2e^{\lambda p_i x}}
       {(1-s_i+s_i e^{\lambda p_i x})^2}\ge0.
\end{equation*}
For $c_i(e)\in[0,1]$, convexity gives $g(c_i(e))\le c_i(e)g(1)$. Therefore,
\begin{equation*}
  \sum_{e\in M}g(c_i(e))
  \le g(1)\cdot c_i(M)
  \le n\cdot g(1).
\end{equation*}
Exponentiating proves \eqref{eq:fractional-load-mgf}.
Next, $\log(1+y)\le y$ for $y\ge0$ implies
\begin{equation*}
  \log\E\left[e^{\lambda p_i c_i(R_i)}\right]
  \le n\log\left(1+s_i(e^{\lambda p_i}-1)\right)
  \le n\cdot s_i\cdot(e^{\lambda p_i}-1).
\end{equation*}
Markov's inequality consequently gives, for every $\lambda>0$,
\begin{equation*}
  \Prb\left[p_i\cdot c_i(R_i)\ge u\right]
  \le\exp\left(-\lambda u+n\cdot s_i\cdot(e^{\lambda p_i}-1)\right).
\end{equation*}
When $u>n\cdot s_i\cdot p_i$, choose
$\lambda=\frac{1}{p_i}\log\frac{u}{n\cdot s_i\cdot p_i}>0$.
The resulting exponent is
\begin{equation*}
  -n\cdot s_i\cdot
  \left(
    \frac{u}{n\cdot s_i\cdot p_i}
    \log\frac{u}{n\cdot s_i\cdot p_i}
    -\frac{u}{n\cdot s_i\cdot p_i}+1
  \right)
  =-n\cdot s_i\cdot\psi\left(\frac{u}{n\cdot s_i\cdot p_i}\right).
\end{equation*}
This proves \eqref{eq:scheduled-chernoff} for $u>n\cdot s_i\cdot p_i$. At equality, its right-hand side is one, so the inequality holds as well.
\end{proofof}

\begin{proofof}{Lemma~\ref{lem:schedule-deficit}}
Throughout this proof, $K=2\log n$. It is convenient to work with the auxiliary values from the backward construction. Recall that
\begin{equation*}
  \hat B_i=\sum_{j=i}^n\hat\beta_j,
  \qquad
  \hat B_i\cdot \psi\left(\frac1{\hat\beta_i}\right)=2\log n,
  \qquad
  \sigma=\frac{n}{\hat B_1}.
\end{equation*}
Hence $n-\frac n\sigma = n-\hat B_1 = \sum_{i=1}^n\left(1-\hat\beta_i\right)$.
We bound each deficit $1-\hat\beta_i$ in terms of the increment $\hat B_i-\hat B_{i+1}$ and then sum these bounds by comparison with elementary integrals.

\paragraph{A pointwise deficit estimate.}
For $y\ge0$, we have $\psi(1+y)=\int_0^y\log(1+t)dt$. Since $\log(1+t)\ge t/(1+t)\ge t/(1+y)$ for $0\le t\le y$, we have $\psi(1+y)\ge\frac{y^2}{2(1+y)}$.

Fix a round $i$ and put $x:=\hat\beta_i$ and $y:=1/x-1=(1-x)/x$. The defining identity gives
\begin{equation*}
  \frac{2\log n}{\hat B_i}
  =\psi(1+y)
  \ge\frac{y^2}{2(1+y)}.
\end{equation*}
Set $a:=2\log n/\hat B_i$. Then $y^2\le2a(1+y)$, so $y\le a+\sqrt{a^2+2a}\le2a+\sqrt{2a}$. Therefore
\begin{equation*}
  \frac{1-\hat\beta_i}{\hat\beta_i}
  \le
  2\sqrt{\frac{\log n}{\hat B_i}}
  +\frac{4\log n}{\hat B_i}.
\end{equation*}
Using $\hat\beta_i=\hat B_i-\hat B_{i+1}$, we obtain
\begin{equation*}
  1-\hat\beta_i
  \le
  2\sqrt{\log n}\cdot
  \frac{\hat B_i-\hat B_{i+1}}{\sqrt{\hat B_i}}
  +4\log n\cdot
  \frac{\hat B_i-\hat B_{i+1}}{\hat B_i}.
\end{equation*}

\paragraph{Summing the deficits.}
We separate the last value and use $1-\hat\beta_n\le1$. For $i\le n-1$, the functions $x^{-1/2}$ and $x^{-1}$ are decreasing, so
\begin{equation*}
  \sum_{i=1}^{n-1}
  \frac{\hat B_i-\hat B_{i+1}}{\sqrt{\hat B_i}}
  \le
  \int_{\hat B_n}^{\hat B_1}\frac{dx}{\sqrt x}
  \le2\sqrt{\hat B_1},
\end{equation*}
and
\begin{equation*}
  \sum_{i=1}^{n-1}
  \frac{\hat B_i-\hat B_{i+1}}{\hat B_i}
  \le
  \int_{\hat B_n}^{\hat B_1}\frac{dx}{x}
  =\log\left(\frac{\hat B_1}{\hat B_n}\right).
\end{equation*}
Consequently,
\begin{equation*}
  n-\hat B_1
  \le
  1+4\sqrt{\hat B_1\log n}
  +4\log n\cdot\log\left(\frac{\hat B_1}{\hat B_n}\right).
\end{equation*}

Every auxiliary value is smaller than one, so $\hat B_1\le n$. At the last backward step, $\hat B_n=\hat\beta_n$, and
\begin{equation*}
  2\log n
  =\hat\beta_n\psi\left(\frac1{\hat\beta_n}\right)
  =\log\left(\frac1{\hat\beta_n}\right)-1+\hat\beta_n.
\end{equation*}
Thus $\hat B_n=\hat\beta_n\ge e^{-1}n^{-2}$. Hence
\begin{equation*}
  n-\frac n\sigma
  =n-\hat B_1
  \le
  1+4\sqrt{n\log n}
  +4\log n\left(3\log n+1\right).
\end{equation*}

The terms other than $4\sqrt{n\log n}$ are $o(\sqrt{n\log n})$. Therefore,
\begin{equation*}
  n-\frac n\sigma
  \le(4+o(1))\sqrt{n\log n}.
\end{equation*}
This proves the first inequality. Dividing it by $n$ gives
\begin{equation*}
  0<1-\frac1\sigma
  \le (4+o(1))\sqrt{\frac{\log n}{n}}.
\end{equation*}
The right-hand side tends to zero, so $\sigma=1+o(1)$. Multiplying the preceding inequality by $\sigma$ gives
\begin{equation*}
  \sigma-1
  \le(4+o(1))\sqrt{\frac{\log n}{n}},
\end{equation*}
which proves the second statement.
\end{proofof}

\section{The Schedule for the Uniform Bound}
\label{app:schedule_for_uniform}

We prove Theorem~\ref{thm:additive-uniform-bound} by giving a different explicit schedule. Let $q=\min\{n,200\}$ and set $B_{n+1}:=0$. For $i\in[n]$, define
\begin{equation*}
  B_i=
  \begin{cases}
    \displaystyle
    \frac{5(n-i+1)^2}{q+4(n-i+1)},
      &n-i+1\le q,\\[6pt]
    n-i+1,
      &n-i+1>q.
  \end{cases}
\end{equation*}
Set $\beta_i:=B_i-B_{i+1}$ and $p_i:=\beta_i/B_i$. Since $B_1=n$, this choice has $s_i=B_i/n$ and therefore defines the schedule used below.

\begin{lemma}\label{lem:uniform-schedule}
The fractions $(p_1,\ldots,p_n)$ form a schedule, and $0<\beta_i\le\frac65$ for every $i\in[n]$.
\end{lemma}

\begin{proof}
For $0\le x\le q$, let $f(x)=5x^2/(q+4x)$. The quotient rule gives
\begin{equation*}
  f'(x)=\frac{10x(q+2x)}{(q+4x)^2},
  \qquad
  f''(x)=\frac{10q^2}{(q+4x)^3}>0.
\end{equation*}
Thus $f'$ is increasing on $[0,q]$, with $f'(0)=0$ and $f'(q)=6/5$.

If $n-i+1\le q$, then
\begin{equation*}
  \beta_i
  =f(n-i+1)-f(n-i)
  =\int_{n-i}^{n-i+1}f'(x)\,dx,
\end{equation*}
so $0<\beta_i\le6/5$. If $n-i+1>q$, then $B_i=n-i+1$ and $B_{i+1}=n-i$, so $\beta_i=1$. Finally, $B_1=n$ and $B_{n+1}=0$, hence $\sum_{i=1}^n\beta_i=n$. Therefore $p_i=\beta_i/B_i$ gives $0<p_i\le1$, and $p_n=1$.
\end{proof}

The following estimates combine the fractional-cost bounds with the centered Bernoulli representation in Lemma~\ref{lem:rounding-cost}.

\begin{lemma}
Let Algorithm~\ref{alg:scheduled-weighted} use the explicit schedule above. For $\lambda>0$, let $\chi(\lambda):=\log((e^\lambda-1)/(e\lambda))+\lambda/(e^\lambda-1)$. Then
\begin{equation}
  \E[\rho(X)]\le\frac1\lambda\cdot\log\left[
    e^\lambda
    +\left(1-\frac{B_n}{n}+\frac{B_n}{n}e^\lambda\right)^n
    +e^{\chi(\lambda)}\cdot\sum_{i=1}^{n-1}
    \left(1-\frac{B_i}{n}
      +\frac{B_i}{n}e^{\lambda\beta_i/B_i}\right)^n\right].
  \label{eq:uniform-softmax}
\end{equation}
For $n\ge200$, we also have
\begin{equation}
  \E[\rho(X)]\le{} \frac52+
    \sum_{i=n-199}^{n}
    \frac{\exp[-B_i\psi((3/2)/\beta_i)]}
         {(B_i/\beta_i)\log((3/2)/\beta_i)}
    +\frac{e^{-201\cdot\psi(3/2)}}
     {201\cdot\log(3/2)\bigl(1-e^{-\psi(3/2)}\bigr)}.
  \label{eq:uniform-integral-certificate}
\end{equation}
\end{lemma}

\begin{proof}
We first bound each agent's exponential moment directly, using the rounding representation from Lemma~\ref{lem:rounding-cost}. We then combine these bounds to prove \eqref{eq:uniform-softmax}, and finally integrate the fractional-cost probability bounds to prove \eqref{eq:uniform-integral-certificate}.
Fix $\lambda>0$ and first consider a round $i<n$. Condition on the remaining item set $R_i$ and the outcomes of all pairwise rounding steps in round $i$, but not on the final Bernoulli decision. By Lemma~\ref{lem:rounding-cost},
\begin{equation*}
  c_i(X_i)=p_i\cdot c_i(R_i)+d\cdot(Z-x),
  \qquad Z\sim\operatorname{Bernoulli}(x),
  \qquad 0\le x\le1,\quad 0\le d\le1.
\end{equation*}
Here $x$ is the value of the last fractional coordinate, $d$ is the corresponding chore's cost, and $Z$ indicates whether that chore is selected in the final rounding step. If no fractional coordinate remains, take $d=x=Z=0$. The bound $d\le1$ follows from the normalization in Section~4. Under this conditioning, $p_i\cdot c_i(R_i),d,x$ are fixed, and only $Z$ can affect the current agent's cost. In the conditional calculations below, $\E_Z$ denotes expectation only over this final Bernoulli decision.

The conditional exponential moment of the error is
\begin{equation*}
  \E_Z[e^{\lambda d(Z-x)}]
  =(1-x)e^{-\lambda dx}+xe^{\lambda d(1-x)}.
\end{equation*}
For fixed $x$, its derivative with respect to $d$ is
\begin{align*}
  \frac{\partial}{\partial d}\bigl[(1-x)e^{-\lambda dx}+xe^{\lambda d(1-x)}\bigr]
  &=-\lambda x(1-x)e^{-\lambda dx}
       +\lambda x(1-x)e^{\lambda d(1-x)}\\
  &=\lambda x(1-x)
       \bigl(e^{\lambda d(1-x)}-e^{-\lambda dx}\bigr)\ge0
       \qquad(d\ge0).
\end{align*}
Thus this exponential moment is nondecreasing in $d$. Since $d\le1$,
\begin{equation*}
  \E_Z[e^{\lambda d(Z-x)}]
  \le e^{-\lambda x}\bigl(1+x(e^\lambda-1)\bigr).
\end{equation*}

Let $g(x):=\log\bigl(1+x(e^\lambda-1)\bigr)-\lambda x$.
A direct derivative calculation shows that $g$ is strictly concave and has its unique maximizer at $x_*:=\frac1\lambda-\frac1{e^\lambda-1}$.
Substitution gives
\begin{equation*}
  g(x_*)
  =\log\left(\frac{e^\lambda-1}{e\lambda}\right)
   +\frac{\lambda}{e^\lambda-1}
  =\chi(\lambda).
\end{equation*}
Therefore $\E_Z[e^{\lambda d(Z-x)}]\le e^{\chi(\lambda)}$ uniformly over all possible $x,d\in[0,1]$.

Since $p_i\cdot c_i(R_i)$ is fixed under the same conditioning,
\begin{equation*}
  \E_Z\left[e^{\lambda c_i(X_i)}\right]
  \le e^{\lambda p_i c_i(R_i)+\chi(\lambda)}.
\end{equation*}
Averaging first over the pairwise-rounding steps and then over $R_i$, and using \eqref{eq:fractional-load-mgf} with $s_i=B_i/n$ and $p_i=\beta_i/B_i$, gives
\begin{equation*}
  \E\left[e^{\lambda c_i(X_i)}\right]
  \le e^{\chi(\lambda)}
       \left(1-\frac{B_i}{n}
                 +\frac{B_i}{n}e^{\lambda\beta_i/B_i}\right)^n.
\end{equation*}
This does not require independence between the rounding error and $p_i\cdot c_i(R_i)$.

For the last round $i=n$, we have $p_n=1$, so every remaining chore is assigned without a random rounding decision and $c_n(X_n)=c_n(R_n)$. Hence $\E\left[e^{\lambda c_n(X_n)}\right]
  \le\left(1-\frac{B_n}{n}+\frac{B_n}{n}e^\lambda\right)^n$.

For every realization, we have $e^{\lambda\rho(X)} \le e^\lambda+\sum_{i=1}^n e^{\lambda c_i(X_i)}$.
Jensen's inequality therefore gives
\begin{equation*}
  \E[\rho(X)]
  \le\frac1\lambda\log\left[
       e^\lambda+\sum_{i=1}^n
             \E[e^{\lambda c_i(X_i)}]\right].
\end{equation*}
Substituting the one-round moment bounds proves \eqref{eq:uniform-softmax}.

Suppose $n\ge200$, so $q=200$, and set $a=3/2$. By Lemma~\ref{lem:uniform-schedule}, every $\beta_i\le6/5<a$. Since $n\cdot s_i\cdot p_i=\beta_i$, the deterministic bound from Section~\ref{sec:additive-load-analysis}, the integral representation of the nonnegative excess, a union bound, and Lemma~\ref{lem:scheduled-load} give
\begin{equation*}
  \E[\rho(X)]\le1+a+\sum_{i=1}^n
            \int_a^\infty e^{-B_i\psi(u/\beta_i)}\,du.
\end{equation*}

For the last $200$ rounds, namely $n-199\le i\le n$, Lemma~\ref{lem:convex-integral} gives
\begin{equation*}
  \int_a^\infty e^{-B_i\psi(u/\beta_i)}\,du
  \le
  \frac{\exp[-B_i\psi(a/\beta_i)]}
       {(B_i/\beta_i)\log(a/\beta_i)}.
\end{equation*}

For the earlier rounds $i\le n-200$, the explicit schedule satisfies $B_i=n-i+1$ and $\beta_i=1$. Therefore
\begin{align*}
  \sum_{i=1}^{n-200}
  \int_a^\infty e^{-B_i\psi(u)}\,du
  \le \sum_{i=1}^{n-200}
  \frac{e^{-(n-i+1)\psi(a)}}{(n-i+1)\log a}
  \le \sum_{k=201}^{\infty}\frac{e^{-k\psi(a)}}{k\log a}
  \le \frac{e^{-201\cdot\psi(a)}}
       {201\cdot\log a\bigl(1-e^{-\psi(a)}\bigr)}.
\end{align*}
For $n=200$, the first sum is empty. Combining the two ranges and substituting $a=3/2$ proves \eqref{eq:uniform-integral-certificate}.
\end{proof}

\begin{proofof}{Theorem~\ref{thm:additive-uniform-bound}}
The schedule is valid by Lemma~\ref{lem:uniform-schedule}.
For $2\le n\le200$, set $\lambda=21/10+\sqrt n/4$ in \eqref{eq:uniform-softmax}. Direct evaluation of the resulting finite certificate gives a maximum below $2.5936<2.594$. For $n\ge200$, the right-hand side of \eqref{eq:uniform-integral-certificate} is independent of $n$ after the geometric suffix is bounded, and its value is below $2.579<2.594$. These two ranges cover every $n\ge2$.
\end{proofof}

The two numerical inequalities used above are routine evaluations of the explicit certificates. Substituting $q=n$ and $\lambda=21/10+\sqrt n/4$ into \eqref{eq:uniform-softmax} for $2\le n\le200$ gives a maximum below $2.593530247$, attained at $n=22$. Evaluating \eqref{eq:uniform-integral-certificate} for the fixed $q=200$ suffix gives a value below $2.578868190$. Both values are strictly below $2.594$, which completes the numerical verification of Theorem~\ref{thm:additive-uniform-bound}.

\section{Supplementary results for goods}\label{app:goods}
In this appendix, we present the results for goods with agent arrival. We use the same agent-arrival model and the same sets $N=[n]$, $M$, $R_i$, and $X_i$ as in the main text. All valuations and the arrival order are fixed before the algorithm makes its random choices. For a nonnegative additive goods valuation $v_i$, define
\[
 \mu_i=\max_{(P_1,\ldots,P_n)\in\Pi_n(M)}\min_{j\in N}v_i(P_j),
 \qquad
 \rho(X)=\min\left\{1,\min_{i:\mu_i>0}\frac{v_i(X_i)}{\mu_i}\right\}.
\]
If all MMS values are zero, set $\rho(X)=1$. We establish the following three results. 
The first two use different parameter choices in one randomized allocation rule. Both algorithms use each arriving agent's MMS and require no prior information about future agents' valuations.

\begin{theorem}[Expected guarantee]\label{thm:goods-expectation}
For every $n$, there is a randomized online algorithm for additive goods such that
\[
 \E[\rho(X)]\ge\frac1{10}.
\]
\end{theorem}

\begin{theorem}[Probability guarantee]\label{thm:goods-probability}
For every $n$ and every $0<x\le1/15$, there is a randomized online algorithm for additive goods such that
\[
 \Prb[\rho(X)\ge x]\ge1-85\cdot x^2.
\]
\end{theorem}

\begin{theorem}[Upper bound]\label{thm:goods-hardness}
For every even $n\ge2$ and every randomized online algorithm, some fixed binary goods instance on $2n$ items satisfies
\[
 \E[\rho(X)]\le\frac12+\frac12\cdot
 \frac{\binom n{n/2}^{\!2}}{\binom{2n}n}
 =\frac12+\frac{1+o(1)}{\sqrt{\pi n}}.
\]
\end{theorem}

We first describe the shared allocation rule and then prove the three results in the order stated.

\subsection{The allocation rule and the current agent's shortfall}\label{subsec:goods-allocation}
The rule has two parameters. A target $a$ specifies when a single good is sufficient for the current agent. A constant $c$ determines the fraction of remaining value allocated when no such good exists. We first describe the rule for $0<a\le1$ and $0<c<2$, and select their values when analyzing the two guarantees.

When agent $i$ arrives, compute her MMS $\mu_i$ for the original set $M$ and the original number $n$ of agents. If $\mu_i>0$, replace each item value by $\min\{v_i(e)/\mu_i,1\}$ and define the value of a bundle as the sum of its processed item values. We continue to write $v_i$ for this processed valuation; $\rho(X)$ always refers to the original instance. Dividing all values by $\mu_i$ makes the MMS one. Capping item values at one preserves an MMS partition: a part containing an item of value at least one still has value at least one after capping, and every other part is unchanged. Conversely, capping cannot increase the minimum value of any partition. The processed MMS is therefore exactly one, and
\begin{equation}\label{eq:goods-normalization}
 0\le v_i(e)\le1,\qquad v_i(M)\ge n,\qquad
 \rho(X)\ge\min\{1,\min_{i:\mu_i>0}v_i(X_i)\}.
\end{equation}

In a nonfinal round $i<n$, set $p_i=c/(n-i+1)$. If $\mu_i=0$ or $v_i(R_i)=0$, allocate nothing. Otherwise, if a remaining good $e$ has processed value at least $a$, allocate the bundle $X_i=\{e\}$, using a fixed item ordering to break ties. A bundle containing exactly one good is called a \emph{singleton}; we refer to this branch as a \emph{singleton allocation}. If no such good exists, run \wpralg{} from Algorithm~\ref{alg:wpr} on the goods with positive value to agent $i$, using weights $\{v_i(e)\}_{e\in R_i}$ and initial fraction $p_i$. Goods of value zero are retained. To ensure all goods are allocated, the final agent receives all remaining goods.

We will account for an agent's loss through her \emph{shortfall}, the amount by which her received value is below $a$. A singleton allocation meets the target $a$ and therefore has no shortfall. In a rounding round, no single remaining good reaches the threshold, so every remaining good has processed value less than $a$. The pair steps preserve the agent's fractional value, and the final Bernoulli decision changes it by less than $a$. The following bound controls the resulting shortfall.
\begin{lemma}\label{lem:goods-shortfall}
For every positive-MMS agent $i<n$, conditional on the history before her arrival,
\begin{equation}\label{eq:goods-shortfall}
 \E[(a-v_i(X_i))_+\mid R_i]
 \le a\cdot\left(1-\frac{p_i\cdot v_i(R_i)}{2a}\right)_+^2,
\end{equation}
where $t_+=\max\{t,0\}$. Moreover, whenever $p_i\cdot v_i(R_i)\ge2a$, agent $i$ receives value at least $a$, regardless of the algorithm's random choices.
\end{lemma}
\begin{proof}
The singleton branch satisfies both claims. In a rounding round, put $\beta=p_i\cdot v_i(R_i)$. By Lemma~\ref{lem:rounding-cost}, conditional on the pair steps, the received value is $\beta+d\cdot(B-q)$, where $0\le d\le a$, $q\in[0,1]$, and $B$ is Bernoulli with mean $q$. If there is no fractional item, take $d=q=0$. The same representation covers an empty allocation.

If $\beta\ge2a$, then $v_i(X_i)\ge\beta-d\ge a$ for every outcome. Otherwise, put $s=\beta/a\in[0,2]$. Conditional on the outcomes of the pairwise rounding steps, the expected shortfall is
\[
\begin{aligned}
\E_B[(a-v_i(X_i))_+]
&=(1-q)\cdot(a-\beta+dq)_+
  +q\cdot(a-\beta-d(1-q))_+\\
&=\max\{0,a-\beta,(1-q)\cdot(a-\beta+dq)\}\\
&\le a\cdot\max\{0,1-s,(1-q)\cdot(1-s+q)\},
\end{aligned}
\]
where $\E_B$ is the expectation over the randomness of the Bernoulli variable $B$ and the second equality holds by the following argument. Since $d\ge0$, we have
\[
(a-\beta+dq)-(a-\beta-d(1-q))=d\ge0.
\]
Thus, $a-\beta-d(1-q)\ge0$ implies $a-\beta+dq\ge0$. If both are nonnegative, then $(1-q)\cdot(a-\beta+dq)_+ +q\cdot(a-\beta-d(1-q))_+ =a-\beta$; if only $a-\beta+dq$ is nonnegative, it equals $(1-q)\cdot(a-\beta+dq)$; if both are negative, it equals zero. The last inequality uses $d\le a$.

Each term in the maximum is at most $(1-s/2)^2$, since
\[
(1-s/2)^2-(1-s)=(s/2)^2\ge0,\qquad
(1-s/2)^2-(1-q)\cdot(1-s+q)=(q-s/2)^2\ge0.
\]
Taking expectation over the pairwise rounding outcomes proves~\eqref{eq:goods-shortfall}.
\end{proof}

\subsection{The expected guarantee}\label{subsec:goods-expectation}
To prove Theorem~\ref{thm:goods-expectation}, we choose $a=1/6$ and $c=1/3$ and bound the sum of the agents' expected shortfalls. The main quantity to control is the value that remains for a future agent. If $k$ agents are still waiting, the difference $k-v_j(R_i)$ measures agent $j$'s shortage relative to one unit per waiting agent. Allocating a singleton decreases the waiting count by one and removes at most one unit of her value, so it cannot increase this shortage. Rounding can increase it, but an exponential function bounds the expected increase in each round. This will suffice to sum all shortfalls directly.

Throughout the analysis, $\mathcal H_i$ denotes the allocations and random choices before round $i$, together with the resulting remaining set $R_i$. For a future agent $j>i$, we first bound $\E[e^{-\lambda\cdot v_j(R_{i+1})}\mid\mathcal H_i]$ in terms of her remaining value $v_j(R_i)$ before round $i$. We will use this bound repeatedly to control the exponential moment when agent $j$ arrives.
\begin{lemma}\label{lem:goods-exponential}
Fix a positive-MMS target agent $j$ and an earlier rounding round $i<j$. For every $\lambda>0$,
\begin{equation}\label{eq:goods-exponential}
 \E[e^{-\lambda\cdot v_j(R_{i+1})}\mid\mathcal H_i]
 \le \exp\left[-\bigl(\lambda-p_i\cdot(e^\lambda-1)\bigr)\cdot v_j(R_i)\right].
\end{equation}
A singleton or empty allocation satisfies $v_j(R_{i+1})\ge v_j(R_i)-1$.
\end{lemma}
\begin{proof}
The proof of Lemma~\ref{lem:rounding-product} also applies when all its factors $z_e$ lie in $(0,1]$: its second-derivative calculation only requires $(z_e-1)\cdot(z_f-1)\ge0$. Set $z_e=e^{-\lambda\cdot v_j(e)}$. An item of zero value to the current agent is retained; including it on the right-hand side is valid because $z_e\le p_i+(1-p_i)\cdot z_e$. Consequently,
\[
 \begin{aligned}
 \E[e^{-\lambda\cdot v_j(R_{i+1})}\mid\mathcal H_i]
 &\le e^{-\lambda\cdot v_j(R_i)}
       \prod_{e\in R_i}\left[1+p_i\cdot(e^{\lambda\cdot v_j(e)}-1)\right]\\
 &\le\exp\left[-\lambda\cdot v_j(R_i)
       +p_i\cdot\sum_{e\in R_i}(e^{\lambda\cdot v_j(e)}-1)\right]\\
 &\le\exp\left[-\bigl(\lambda-p_i\cdot(e^\lambda-1)\bigr)\cdot v_j(R_i)\right].
 \end{aligned}
\]
The second inequality uses $1+t\le e^t$. The last uses convexity on $[0,1]$, which implies $e^{\lambda\cdot v_j(e)}-1\le v_j(e)\cdot(e^\lambda-1)$. Finally, a singleton has target value at most one by~\eqref{eq:goods-normalization}.
\end{proof}

For the expected guarantee, one fixed exponential is enough. The next estimate shows that its expectation stays bounded throughout the allocation, uniformly over the target agent and the number of agents.

\begin{lemma}\label{lem:goods-expected-moment}
Use $c=1/3$ and let $\lambda=\log(11/2)$. For a positive-MMS target agent $j$, put $\ell=n-j+1$. Then
\begin{equation}\label{eq:goods-expected-moment}
 \E[e^{\lambda\cdot(\ell-v_j(R_j))}]
 \le
 \begin{cases}
  11/8,&\ell=1,\\
  49/40,&\ell\ge2.
 \end{cases}
\end{equation}
\end{lemma}
\begin{proof}
Initially, the exponential is at most one because $v_j(M)\ge n$. A singleton or empty allocation cannot increase it. Consider a rounding round with $k=n-i+1\ge2$, and condition on $\mathcal H_i$. By~\eqref{eq:goods-exponential}, its expected increase is at most
\begin{equation}\label{eq:goods-moment-increase}
 e^{\lambda\cdot(k-1)-(\lambda-3/(2k))\cdot v_j(R_i)}
 -e^{\lambda\cdot(k-v_j(R_i))},
\end{equation}
where we used $p_i=1/(3k)$ and $e^\lambda-1=9/2$.

We bound this increase by a quantity depending only on $k$. If~\eqref{eq:goods-moment-increase} is nonpositive, the desired upper bound follows immediately. Otherwise, dividing the positive difference by $e^{\lambda\cdot(k-v_j(R_i))}$, we have
\[
e^{-\lambda+(3/(2k))\cdot v_j(R_i)}>1,
\]
and hence
\[
v_j(R_i)>\frac{2\lambda k}{3},
\qquad
0<e^{\lambda-(3/(2k))\cdot v_j(R_i)}<1.
\]

Since this is a nonfinal round, $k\ge2$. We apply the arithmetic--geometric mean inequality to $k-1$ copies of $e^{\lambda-(3/(2k))\cdot v_j(R_i)}/(k-1)$ and one copy of $1-e^{\lambda-(3/(2k))\cdot v_j(R_i)}$. These $k$ numbers are nonnegative and sum to $1$, so
\[
\left(\frac{e^{\lambda-(3/(2k))\cdot v_j(R_i)}}{k-1}\right)^{k-1}
\cdot\left(1-e^{\lambda-(3/(2k))\cdot v_j(R_i)}\right)
\le\frac{1}{k^k}.
\]
Multiplying by $(k-1)^{k-1}$, we obtain
\[
e^{(k-1)\cdot(\lambda-(3/(2k))\cdot v_j(R_i))}
\cdot\left(1-e^{\lambda-(3/(2k))\cdot v_j(R_i)}\right)
\le\frac{(k-1)^{k-1}}{k^k}.
\]

For $\lambda=\log(11/2)$, we have $\lambda>3/2$ and $e^{-\lambda\cdot(2\lambda/3-1)}<4/5$. Since $v_j(R_i)>2\lambda k/3$, it follows that
\[
e^{-(\lambda-3/2)\cdot v_j(R_i)}
\le e^{-\lambda\cdot(2\lambda/3-1)\cdot k}
\le\left(\frac45\right)^k.
\]
Combining these bounds, we have
\[
\begin{aligned}
&e^{\lambda\cdot(k-1)-(\lambda-3/(2k))\cdot v_j(R_i)}
-e^{\lambda\cdot(k-v_j(R_i))}\\
&\quad=
e^{-(\lambda-3/2)\cdot v_j(R_i)}
\cdot e^{(k-1)\cdot(\lambda-(3/(2k))\cdot v_j(R_i))}
\cdot\left(1-e^{\lambda-(3/(2k))\cdot v_j(R_i)}\right)\\
&\quad\le
\left(\frac45\right)^k\cdot\frac{(k-1)^{k-1}}{k^k}.
\end{aligned}
\]

Let $\eta_k=(4/5)^k\cdot(k-1)^{k-1}/k^k$ denote this bound on the increase in a round with $k$ waiting agents. Taking expectations and summing over earlier rounds yields
\begin{equation}\label{eq:goods-moment-sum}
 \E[e^{\lambda\cdot(\ell-v_j(R_j))}]
 \le1+\sum_{k=\ell+1}^{\infty}\eta_k.
\end{equation}
To bound the series, use $(1-1/k)^k\le e^{-1}$ and keep its first four terms:
\[
 \sum_{k=2}^{\infty}\eta_k
 \le\sum_{k=2}^{5}\left(\frac45\right)^k\frac{(k-1)^{k-1}}{k^k}
 +\frac{4}{5e}\left[\log5-\sum_{t=1}^{4}\frac{(4/5)^t}{t}\right]
 <\frac38.
\]
For the first inequality, set $t=k-1$ in the tail and use $\sum_{t\ge1}(4/5)^t/t=-\log(1-4/5)=\log5$. Since $\eta_2=4/25$, we also have $\sum_{k=3}^{\infty}\eta_k<9/40$. Combining these two estimates with~\eqref{eq:goods-moment-sum} proves~\eqref{eq:goods-expected-moment}.
\end{proof}

We can now finish the expectation analysis by converting the bounded exponential into shortfall bounds. Earlier agents have a quadratic shortfall divided by the square of the number still waiting, so their total is summable. The final agent receives everything that remains and is bounded separately.

\begin{proof}[Proof of Theorem~\ref{thm:goods-expectation}]
The case $n=1$ has ratio one. Assume $n\ge2$. For a positive-MMS agent $j<n$, write $\ell=n-j+1\ge2$. Since $a=1/6$ and $p_j=1/(3\ell)$, Lemma~\ref{lem:goods-shortfall} implies
\[
 \E[(a-v_j(X_j))_+]
 \le\frac1{6\ell^2}\cdot\E[(\ell-v_j(R_j))_+^2]
 \le\frac{49}{40}\cdot\frac{2}{3e^2\lambda^2\ell^2}.
\]
For the second inequality, apply~\eqref{eq:goods-expected-moment} and $(t_+)^2\le4e^{\lambda t}/(e^2\lambda^2)$; the latter follows by maximizing $t^2e^{-\lambda t}$ at $t=2/\lambda$ for $t\ge0$.

For the final agent, $\lambda<2$, so $\lambda a<1$. For $0\le w<a$, the inequality $e^{-\lambda w}\ge1-\lambda w\ge1-w/a$ implies $(a-w)_+\le a\cdot e^{-\lambda w}$; the same bound is immediate for $w\ge a$. Therefore, if $\mu_n>0$,
\[
 \E[(a-v_n(X_n))_+]
 \le a\cdot e^{-\lambda}\cdot
       \E[e^{\lambda\cdot(1-v_n(R_n))}]
 \le\frac16\cdot\frac2{11}\cdot\frac{11}{8}
 =\frac1{24}.
\]
If $\mu_n=0$, there is no final-agent shortfall to include.

To sum the other agents' shortfalls, observe that
\[
 \sum_{\ell=2}^{\infty}\frac1{\ell^2}
 \le\frac14+\sum_{\ell=3}^{\infty}
 \left(\frac1{\ell-1/2}-\frac1{\ell+1/2}\right)
 =\frac14+\frac25=\frac{13}{20}.
\]
The comparison holds because the parenthesized difference is $1/(\ell^2-1/4)\ge1/\ell^2$. Consequently,
\[
 \sum_{\substack{j<n\\\mu_j>0}}\E[(a-v_j(X_j))_+]
 \le\frac{49}{40}\cdot\frac{2}{3e^2\lambda^2}\cdot\frac{13}{20}
 <\frac1{40}.
\]
Finally, \eqref{eq:goods-normalization} implies the pointwise bound
\[
 (a-\rho(X))_+
 \le\sum_{j:\mu_j>0}(a-v_j(X_j))_+.
\]
Indeed, if the realized ratio is below $a$, its shortfall is at most the largest processed-value shortfall, which is at most their sum. Taking expectations gives
\[
 \E[\rho(X)]\ge a-\E[(a-\rho(X))_+]
 \ge\frac16-\frac1{24}-\frac1{40}=\frac1{10},
\]
which completes the proof.
\end{proof}

\subsection{The probability guarantee}\label{subsec:goods-probability}
To prove Theorem~\ref{thm:goods-probability}, fix $0<x\le1/15$ and choose $a=x$ and $c=3x$. For a nonfinal agent $i$, we have $a=x$ and $p_i=3x/(n-i+1)$. If $v_i(R_i)\ge2(n-i+1)/3$, then
\[
p_i\cdot v_i(R_i)
\ge\frac{3x}{n-i+1}\cdot\frac{2(n-i+1)}3
=2x=2a.
\]
By Lemma~\ref{lem:goods-shortfall}, agent $i$ therefore receives value at least $a=x$. We will bound the probability of falling below these thresholds. The exponential estimate alone loses too much near a threshold, so we use an explicit function that measures the shortage linearly and decays geometrically when there is a surplus.

\subsubsection{Comparing one allocation with independent deletions}
To use a function other than an exponential, we need a convex version of Lemma~\ref{lem:goods-exponential}. The argument follows the same pair-step concavity calculation as Lemma~\ref{lem:rounding-product}. It first compares weighted rounding with independently deleting each item, and then replaces each deleted fractional value by a zero-or-one contribution of the same mean.

\begin{lemma}\label{lem:goods-convex}
Fix a positive-MMS target agent $j$ and an earlier rounding round $i<j$. Conditional on $\mathcal H_i$, let $D$ be a sum of independent Bernoulli variables, one for each $e\in R_i$, with respective means $p_i\cdot v_j(e)$. For every convex nonincreasing function $g$ on $\mathbb R$ and every fixed shift $s$,
\begin{equation}\label{eq:goods-convex}
 \E[g(v_j(R_{i+1})+s)\mid\mathcal H_i]
 \le\E[g(v_j(R_i)+s-D)],
 \qquad \E[D]=p_i\cdot v_j(R_i).
\end{equation}
\end{lemma}
\begin{proof}
Consider the expected loss when items are independently deleted with probabilities equal to their current rounding coordinates. We express the independent-deletion comparison in terms of retention probabilities: in the following calculation, $x_e$ and $x_f$ denote the complements of the corresponding allocation coordinates. Thus, items $e$ and $f$ are independently deleted with probabilities $1-x_e$ and $1-x_f$. Fix the deletion outcomes of all other items, and let $b$ be their retained value plus the fixed shift $s$. The four possible outcomes for $e,f$ yield the conditional expectation
\[
\begin{aligned}
&(1-x_e)\cdot(1-x_f)\cdot g(b)\\
&\quad+x_e\cdot(1-x_f)\cdot g(b+v_j(e))\\
&\quad+(1-x_e)\cdot x_f\cdot g(b+v_j(f))\\
&\quad+x_e\cdot x_f\cdot g(b+v_j(e)+v_j(f)).
\end{aligned}
\]
Expanding the probability products and collecting terms, this expectation equals
\[
\begin{aligned}
&g(b)
+x_e\cdot\bigl[g(b+v_j(e))-g(b)\bigr]\\
&\quad+x_f\cdot\bigl[g(b+v_j(f))-g(b)\bigr]\\
&\quad+x_e\cdot x_f\cdot
\bigl[g(b)-g(b+v_j(e))-g(b+v_j(f))
+g(b+v_j(e)+v_j(f))\bigr].
\end{aligned}
\]
The bracket in the last line is therefore the coefficient of $x_e\cdot x_f$.

We next show that this coefficient is nonnegative. If $v_j(e)+v_j(f)=0$, the coefficient is zero. Otherwise, both $b+v_j(e)$ and $b+v_j(f)$ lie between $b$ and $b+v_j(e)+v_j(f)$. By convexity of $g$, we have
\[
\begin{aligned}
g(b+v_j(e))
&\le
\frac{v_j(f)}{v_j(e)+v_j(f)}\cdot g(b)
+\frac{v_j(e)}{v_j(e)+v_j(f)}
\cdot g(b+v_j(e)+v_j(f)),\\
g(b+v_j(f))
&\le
\frac{v_j(e)}{v_j(e)+v_j(f)}\cdot g(b)
+\frac{v_j(f)}{v_j(e)+v_j(f)}
\cdot g(b+v_j(e)+v_j(f)).
\end{aligned}
\]
Adding these inequalities, the coefficients of each endpoint value sum to one, so
\[
g(b+v_j(e))+g(b+v_j(f))
\le g(b)+g(b+v_j(e)+v_j(f)).
\]
Consequently,
\[
g(b)-g(b+v_j(e))-g(b+v_j(f))
+g(b+v_j(e)+v_j(f))\ge0.
\]
We now use the nonnegative coefficient to analyze a weighted pair step. Along this step, the retention coordinates $x_e,x_f$ change with slopes $-1/v_i(e)$ and $1/v_i(f)$, respectively. The constant and linear terms in the expanded expectation have zero second derivative along the step, while the product $x_e\cdot x_f$ has second derivative $-2/(v_i(e)\cdot v_i(f))$. Therefore, the second derivative of the displayed expectation is
\[
-\frac{2}{v_i(e)\cdot v_i(f)}
\cdot\bigl[g(b)-g(b+v_j(e))-g(b+v_j(f))
+g(b+v_j(e)+v_j(f))\bigr]
\le0.
\]
This expectation is thus concave along the pair step. Averaging over the deletion outcomes of the other items preserves this concavity. Conditional on the current rounding coordinates and the outcomes of the preceding pair steps, each coordinate preserves its expectation. For a vector of retention probabilities $\boldsymbol{x}=(x_e)_{e\in R_i}$, define
\[
F(\boldsymbol{x})
=\E\left[g\left(s+\sum_{e\in R_i}v_j(e)\cdot Y_e\right)\right],
\qquad
Y_e\sim\operatorname{Bernoulli}(x_e)
\quad\text{independently}.
\]
Thus, $F(\boldsymbol{x})$ is the expected value of $g$ when each item $e$ is independently retained with probability $x_e$. Condition on $\mathcal H_i$ and the outcomes of the preceding pair steps, so the current vector $\boldsymbol{x}$ is fixed. Let $\boldsymbol{x}'$ be the random vector after the next pair step, and let $\E_{\mathrm{pair}}$ denote expectation over this step's random choice. Since each coordinate preserves its conditional expectation,
\[
\E_{\mathrm{pair}}[\boldsymbol{x}']=\boldsymbol{x}.
\]
By the preceding calculation, $F$ is concave along the line on which this pair step moves. Therefore, by Jensen's inequality,
\[
\E_{\mathrm{pair}}[F(\boldsymbol{x}')]
\le F\bigl(\E_{\mathrm{pair}}[\boldsymbol{x}']\bigr)
=F(\boldsymbol{x}).
\]

After all pair steps, fix their outcomes. If $x_e$ is the remaining fractional coordinate, we have
\[
b=s+\sum_{e'\in R_i\setminus\{e\}}v_j(e')\cdot x_{e'}.
\]
All coordinates in this sum are integral. The final choice deletes $e$ with probability $1-x_e$ and retains it with probability $x_e$. Hence, taking expectation only over this final choice, we have
\[
\E_{\mathrm{final}}\left[g\left(v_j(R_{i+1})+s\right)\right]
=(1-x_e)\cdot g(b)+x_e\cdot g(b+v_j(e))
=F(\boldsymbol{x}).
\]
If all coordinates are already integral, no final random choice is needed, and directly
\[
g\left(v_j(R_{i+1})+s\right)=F(\boldsymbol{x}).
\]

Let $\boldsymbol{x}$ now denote the initial retention vector: $x_e=1-p_i$ if $v_i(e)>0$, and $x_e=1$ otherwise. Applying the pair-step inequality successively and then averaging over the final choice bounds the actual expectation by $F$ at the initial retention probabilities. These probabilities are $1-p_i$ for items with positive value to agent $i$, and $1$ for items with zero value to her. Consequently,
\[
\E\left[g\left(v_j(R_{i+1})+s\right)\mid\mathcal H_i\right]
\le F(\boldsymbol{x})
,
\]
where the variables $Y_e$ are independent and satisfy
\[
\Pr[Y_e=1]=
\begin{cases}
1-p_i,&v_i(e)>0,\\
1,&v_i(e)=0.
\end{cases}
\]

Draw mutually independent variables $\xi_e\sim\operatorname{Bernoulli}(p_i)$ and $\zeta_e\sim\operatorname{Bernoulli}(v_j(e))$, and put $D=\sum_{e\in R_i}\xi_e\cdot\zeta_e$. Conditional on $(\xi_e)_e$, we have $\E[D\mid(\xi_e)_e]=\sum_{e\in R_i}v_j(e)\cdot\xi_e$. By convexity of $g$ and Jensen's inequality,
\[
\begin{aligned}
g\left(v_j(R_i)+s-\sum_{e\in R_i}v_j(e)\cdot\xi_e\right)
&=g\left(v_j(R_i)+s-\E[D\mid(\xi_e)_e]\right)\\
&\le\E\left[g\left(v_j(R_i)+s-D\right)\mid(\xi_e)_e\right].
\end{aligned}
\]

In the initial independent comparison, only items with $v_i(e)>0$ can be deleted. Since $g$ is non-increasing, deleting items with $v_i(e)=0$ can only increase its value. Using this observation and then taking expectation in the preceding Jensen inequality, we obtain
\[
\begin{aligned}
F(\boldsymbol{x})
&=\E_\xi\left[g\left(v_j(R_i)+s-
\sum_{\substack{e\in R_i\\v_i(e)>0}}v_j(e)\cdot\xi_e\right)\right]\\
&\le\E_\xi\left[g\left(v_j(R_i)+s-
\sum_{e\in R_i}v_j(e)\cdot\xi_e\right)\right]\\
&\le\E\left[g\left(v_j(R_i)+s-D\right)\right].
\end{aligned}
\]
Combining this bound with the comparison established by the rounding steps, we have
\[
\E\left[g\left(v_j(R_{i+1})+s\right)\mid\mathcal H_i\right]
\le F(\boldsymbol{x})
\le\E\left[g\left(v_j(R_i)+s-D\right)\right].
\]
The products $\xi_e\cdot\zeta_e$ are independent Bernoulli variables with means $p_i\cdot v_j(e)$, so $\E[D]=p_i\cdot v_j(R_i)$. This proves~\eqref{eq:goods-convex}. The argument $v_j(R_i)+s-D$ may be negative, hence $g$ is defined on all of $\mathbb R$.
\end{proof}

\subsubsection{An explicit function measuring shortage}
To bound the events $v_j(R_j)<2(n-j+1)/3$, we construct a potential $H$ that measures how little value remains for a future agent. Fix a future agent $j$ and consider round $i<j$, with $k=n-i+1$ agents still waiting. We evaluate $H$ at $v_j(R_i)-k$, the remaining value minus one unit per waiting agent. We require $H$ to be a nonnegative, non-increasing, and convex function, with $H(t)\ge\max\{-t,0\}$. In particular, if $v_j(R_i)<k$, then $H(v_j(R_i)-k)\ge k-v_j(R_i)$: the potential is at least the missing value.

Allocating $X_i$ removes $v_j(X_i)$ units of agent $j$'s value and reduces the waiting count by one. Thus,
\[
v_j(R_{i+1})-(k-1)
=v_j(R_i)-k+1-v_j(X_i).
\]
For a rounding round, Lemma~\ref{lem:goods-convex} therefore implies
\[
\E\left[H\left(v_j(R_{i+1})-(k-1)\right)\mid\mathcal H_i\right]
\le
\E\left[H\left(v_j(R_i)-k+1-D\right)\right],
\qquad
\E[D]=p_i\cdot v_j(R_i).
\]

We will choose $H$ so that, for every real $t$,
\[
\E[H(t+1-D)]\le H(t)
\qquad\text{whenever }\E[D]\le\frac92\cdot x.
\]
Taking $t=v_j(R_i)-k$ then shows that the expected potential cannot increase in any rounding round satisfying this mean bound. We will later apply the same comparison with a fixed shift in the argument of $H$ to bound each agent's failure probability.

In a round with $k$ waiting agents, $p_i=3x/k$. When the target's remaining value is at most $3k/2$, the mean comparison count is at most
\[
 \nu=\frac92\cdot x\le\frac3{10}.
\]
We first control this range. On negative inputs, let $H$ be linear so it records the shortage. On nonnegative integers, let it decrease geometrically. At input zero, the next input is negative only if at least two units are deleted, which motivates a starting height proportional to $\nu^2$. The following choices make both requirements work:
\begin{equation}\label{eq:goods-H}
 h=\frac{13}{18}\cdot\nu^2,\qquad z=\frac{\nu}{2},\qquad
 H(k)=
 \begin{cases}
  h-k,&k\le0,\\
  h\cdot z^k,&k\ge0,
 \end{cases}
 \quad k\in\mathbb Z.
\end{equation}
We extend $H$ to real inputs as follows. For $t\le0$, set $H(t)=h-t$. For every integer $k\ge0$ and $t\in[k,k+1]$, define
\[
H(t)=(k+1-t)\cdot h z^k+(t-k)\cdot h z^{k+1}.
\]
The two coefficients are nonnegative and sum to one. At $t=k$ and $t=k+1$, this formula equals $H(k)$ and $H(k+1)$, respectively, and between them it is linear in $t$.

On the negative half-line, the slope is $-1$. On $[k,k+1]$, the interval has length one, so its slope is
\[
H(k+1)-H(k)
=h z^{k+1}-h z^k
=-h z^k(1-z).
\]
Since $0<z<1$ and $h(1-z)\le h\le13/200<1$, these slopes satisfy
\[
-1\le-h(1-z)\le-hz(1-z)\le-hz^2(1-z)\le\cdots<0.
\]
All slopes are nonpositive, so $H$ is nonincreasing. The slopes never decrease as the input increases, so the continuous piecewise linear function $H$ is convex.

Finally, for $t\le0$, we have $H(t)=h-t\ge-t$. For $t\ge0$, choose an integer $k\ge0$ with $k\le t\le k+1$. Then $k+1-t\ge0$ and $t-k\ge0$. Since $h>0$ and $z>0$, we have
\[
(k+1-t)\cdot hz^k\ge0,
\qquad
(t-k)\cdot hz^{k+1}\ge0.
\]
Therefore,
\[
H(t)=(k+1-t)\cdot hz^k+(t-k)\cdot hz^{k+1}\ge0.
\]
Consequently, $H$ is nonnegative and satisfies $H(t)\ge\max\{-t,0\}=(-t)_+$ for every real $t$.
We now verify that this choice of $H$ satisfies $\E[H(t+1-D)]\le H(t)$ whenever $\E[D]\le\nu$. To handle rounds with larger remaining values, we also derive an exponential upper bound on $H$ in the following lemma, which allows us to apply Lemma~\ref{lem:goods-exponential}. For this bound, set $\lambda=\log(1/z)$, so $z^t=e^{-\lambda t}$.

\begin{lemma}\label{lem:goods-H}
For every sum $D$ of independent Bernoulli variables with $\E[D]\le\nu$ and every real $t$,
\begin{equation}\label{eq:goods-H-bounds}
 \E[H(t+1-D)]\le H(t),\qquad
 H(t)\le\frac{e^{\lambda h}}{e\lambda}\cdot z^t.
\end{equation}
\end{lemma}
\begin{proof}
We first prove the expected-decrease inequality at integer inputs. For $k\le-1$, every input $k+1-D$ is nonpositive, so
$\E[H(k+1-D)]-H(k)=\E[D]-1\le0$. For $k\ge0$, write $\delta_k=\E[H(k+1-D)]-H(k)$ for this expected change. At zero, separating $D=0$ from $D\ge1$ gives
\[
 \delta_0=\E[D]-1+[1-h(1-z)]\cdot\Prb[D=0]
 \le\nu-1+[1-h(1-z)]\cdot e^{-\nu}.
\]
Indeed, independence implies $\Prb[D=0]\le e^{-\E[D]}$, and the function $u-1+[1-h(1-z)]e^{-u}$ is increasing for $u\ge0$ because its derivative is $1-[1-h(1-z)]e^{-u}>0$. Thus $\delta_0\le0$ follows from $h(1-z)\ge1-(1-\nu)e^\nu$. The exponential series verifies this:
\[
 1-(1-\nu)e^\nu
 \le\frac{\nu^2}{2}+\frac{\nu^3}{3}+\frac{\nu^4}{8}+\frac{\nu^5}{24},
\]
\[
 h(1-z)-\left(\frac{\nu^2}{2}+\frac{\nu^3}{3}+\frac{\nu^4}{8}+\frac{\nu^5}{24}\right)
 =\nu^2\left(\frac29-\frac{25\nu}{36}-\frac{\nu^2}{8}-\frac{\nu^3}{24}\right)
 >0.
\]
For the first line, multiply $e^\nu\ge1+\nu+\nu^2/2+\nu^3/6+\nu^4/24$ by $1-\nu>0$. For the last inequality, the parenthesis decreases with $\nu$ and is positive at $\nu=3/10$.

To extend this check to positive integers, subtract $z\cdot\delta_{k-1}$ from $\delta_k$. The geometric branch cancels and, for $k\ge1$,
\begin{equation}\label{eq:goods-H-residual}
 \delta_k-z\delta_{k-1}
 =\sum_{d\ge k+1}\Prb[D=d]\cdot[(1-z)(h+d-k)-1].
\end{equation}
To see the cancellation, use $H(u+1)-zH(u)=0$ for integer $u\ge0$ and $H(u+1)-zH(u)=(1-z)(h-u)-1$ for integer $u<0$, with $u=k-D$.

We bound the excess deletion in~\eqref{eq:goods-H-residual} directly. Since $D=\sum_e\xi_e$ with independent Bernoulli $\xi_e$, for every integer $q\ge0$, independence implies
\[
 \begin{aligned}
 (q+1)\cdot\Prb[D\ge q+1]
 &\le\E[D\cdot\mathbf1_{\{D\ge q+1\}}] =\sum_e\E[\xi_e]\cdot\Prb[D-\xi_e\ge q]
 \le\nu\cdot\Prb[D\ge q].
 \end{aligned}
\]
In the equality, a summand is nonzero only when $\xi_e=1$, and $D-\xi_e$ is independent of $\xi_e$. Iterating this bound and summing a geometric series yields, whenever $\Prb[D\ge q]>0$ and $q\ge1$,
\[
 \E[D-q\mid D\ge q]
 =\sum_{b\ge1}\frac{\Prb[D\ge q+b]}{\Prb[D\ge q]}
 \le\sum_{b\ge1}\left(\frac{\nu}{q+1}\right)^b
 =\frac{\nu}{q+1-\nu}.
\]
If $\Prb[D\ge k+1]=0$, the right-hand side of~\eqref{eq:goods-H-residual} vanishes. Otherwise, apply this bound with $q=k+1$ to obtain
\[
 \delta_k-z\delta_{k-1}
 \le\Prb[D\ge k+1]\cdot
 \left[(1-z)\left(h+1+\frac{\nu}{3-\nu}\right)-1\right]\le0,
\]
where the last inequality is equivalent to $h\le\nu/[(2-\nu)(3-\nu)]$, or $(13/18)\cdot\nu(2-\nu)(3-\nu)\le1$. The polynomial $\nu(2-\nu)(3-\nu)$ is increasing on this interval, since its derivative is $6-10\nu+3\nu^2>0$, and the required bound holds at $\nu=3/10$. By induction, we have $\delta_k\le0$ for every nonnegative integer $k$. It remains to prove $\E[H(t+1-D)]\le H(t)$ for non-integer inputs. Write $t=k+\theta$, where $k\in\mathbb Z$ and $0\le\theta<1$. Since $D$ is integer-valued, $t+1-D$ lies between the consecutive integers $k+1-D$ and $k+2-D$. By the definition of $H$ between integers,
\[
H(t+1-D)
=(1-\theta)\cdot H(k+1-D)
+\theta\cdot H(k+2-D).
\]
Taking the expectation, and applying the already proved integer-input inequality at $k$ and $k+1$, we have
\[
\begin{aligned}
\E[H(t+1-D)]
&=(1-\theta)\cdot\E[H(k+1-D)]
+\theta\cdot\E[H(k+2-D)]\\
&\le(1-\theta)\cdot H(k)+\theta\cdot H(k+1)\\
&=H(k+\theta)=H(t).
\end{aligned}
\]
The last equality again follows from the definition of $H$ between integers. Thus $\E[H(t+1-D)]\le H(t)$ holds for every real $t$.

For the upper bound, maximizing $y\cdot e^{-\lambda y}$ over $y\ge0$ gives $y\cdot e^{-\lambda y}\le1/(e\cdot\lambda)$. If $t<0$, we use $y=h-t$ and obtain
\[
 H(t)\cdot z^{-t}=e^{\lambda h}\cdot(h-t)e^{-\lambda(h-t)}
 \le\frac{e^{\lambda h}}{e\cdot\lambda}.
\]
If $t=k+\theta\ge0$, with integer $k\ge0$ and $0\le\theta<1$, the definition between consecutive integers gives
\[
 H(t)\cdot z^{-t}
 =h\cdot[(1-\theta)+\theta z]\cdot z^{-\theta}
 \le h\cdot\left(\frac1{e\lambda z}+1\right).
\]
Here $(1-\theta)\cdot z^{-\theta}=z^{-1}\cdot (1-\theta)\cdot e^{-\lambda(1-\theta)}\le1/(e\cdot\lambda \cdot z)$ and $\theta \cdot z^{1-\theta}\le1$. Since $\nu=2e^{-\lambda}$, the same scalar exponential bound gives $e\cdot\nu\cdot \lambda\le2$. Hence
\[
 e\cdot \lambda \cdot h+h/z
 =\frac{13}{18}\cdot\nu\cdot (e\nu\lambda+2)
 \le\frac{26}{9}\cdot\nu\le\frac{13}{15}<1\le e^{\lambda h}.
\]
Multiplying the preceding bound by $e\lambda$ completes the proof.
\end{proof}

From the last step in the proof, we also have $\nu\lambda\le2/e<1$, which will be useful below.

\subsubsection{Following the potential until the target arrives}
In this section, we track $H(v_j(R_i)-k+m)$, where $k$ agents are still waiting and $m\ge0$ is a fixed integer. If $v_j(R_i)\le3k/2$, then by Lemmas~\ref{lem:goods-convex} and~\ref{lem:goods-H}, the conditional expected potential cannot increase. Otherwise, $v_j(R_i)-k>k/2$, and the exponential estimate bounds the expected increase by a term that decays geometrically in $k$. We sum these terms until agent $j$ arrives, choosing $m$ later to match her failure threshold.

\begin{lemma}\label{lem:goods-probability-moment}
For every positive-MMS target agent $j$, put $\ell=n-j+1$. For every integer $m\ge0$,
\begin{equation}\label{eq:goods-probability-moment}
 \E[H(v_j(R_j)-\ell+m)]
 \le H(m)+\frac{11}{9}\cdot
       \frac{z^{m+(\ell+3)/2}}{1-\sqrt z}.
\end{equation}
\end{lemma}
\begin{proof}
Condition on $\mathcal H_i$ in an earlier round, and write $k=n-i+1\ge2$. A singleton or empty allocation cannot increase $H(v_j(R_i)-k+m)$. If rounding is used and $v_j(R_i)\le3k/2$, the comparison count has mean $p_i\cdot v_j(R_i)\le\nu$. By~\eqref{eq:goods-convex} and~\eqref{eq:goods-H-bounds},
\[
 \E[H(v_j(R_{i+1})-(k-1)+m)\mid\mathcal H_i]
 \le\E[H(v_j(R_i)-k+m+1-D)]
 \le H(v_j(R_i)-k+m).
\]
If $v_j(R_i)>3k/2$, apply the exponential part of~\eqref{eq:goods-H-bounds} and then~\eqref{eq:goods-exponential}. The coefficient $\lambda-p_i(1/z-1)$ is positive: since $p_i=2\nu/(3k)\le\nu/3$, we have $p_i(1/z-1)\le(2-\nu)/3<1$, whereas $\lambda\ge\log(20/3)>1$. Thus we may replace $v_j(R_i)$ by $3k/2$ in the decreasing exponential and obtain
\[
 \begin{aligned}
 \E[H(v_j(R_{i+1})-(k-1)+m)\mid\mathcal H_i]
 &\le\frac{e^{\lambda h}}{e\lambda}
 e^{\lambda(k-1-m)-(\lambda-p_i(1/z-1))\cdot(3k/2)}\\
 &=\frac{e^{\lambda h+1-\nu}}{\lambda}\cdot z^{m+k/2+1},
 \end{aligned}
\]
where the equality holds by $p_i(3k/2)=\nu$ and $\nu(1/z-1)=2-\nu$.

We next check a uniform upper bound of $11/9$ on this coefficient. Since $\lambda=\log(2/\nu)$,
\[
 \frac{\mathrm d}{\mathrm d\nu}
 \log\left(\frac{e^{\lambda h+1-\nu}}{\lambda}\right)
 =\frac{13}{18}\cdot\nu(2\lambda-1)-1+\frac1{\nu\lambda}>0.
\]
The sign follows from $\lambda>1$ and $\nu\lambda<1$. Its maximum therefore occurs at $\nu=3/10$, where $h=13/200$ and $\lambda=\log(20/3)$. Substituting this endpoint, we obtain
\[
 \frac{e^{\lambda h+1-\nu}}{\lambda}
 \le\frac{\exp[(13/200)\cdot\log(20/3)+7/10]}{\log(20/3)}
 <\frac{11}{9}.
\]

The potential before round $i$ satisfies $H(v_j(R_i)-k+m)\ge0$, so all branches satisfy
\[
 \E[H(v_j(R_{i+1})-(k-1)+m)\mid\mathcal H_i]
 \le H(v_j(R_i)-k+m)+\frac{11}{9}\cdot z^{m+k/2+1}.
\]
Initially, $H(v_j(M)-n+m)\le H(m)$. Take expectations and sum over earlier waiting counts $k=n,n-1,\ldots,\ell+1$. Since $0<z<1$, we have
\[
\sum_{k=\ell+1}^{n}z^{m+k/2+1}
\le z^{m+(\ell+3)/2}\sum_{k=0}^{\infty}(\sqrt z)^k
=\frac{z^{m+(\ell+3)/2}}{1-\sqrt z}.
\]
Substituting this bound into the preceding inequality, we complete the proof.
\end{proof}

It remains to translate the potential bound into a bound on the probability that some agent fails. We group the number of waiting agents into triples because the sufficient remaining value is two thirds of that number. Within each triple, one common integer shift makes all three failure thresholds negative, where $H$ directly measures a shortage.

\begin{proof}[Proof of Theorem~\ref{thm:goods-probability}]
The case $n=1$ is immediate. Suppose $n\ge2$. By~\eqref{eq:goods-normalization} and the union bound, we have
\[
\Prb[\rho(X)<x]
\le\sum_{i:\mu_i>0}\Prb[v_i(X_i)<x].
\]
We now bound each probability on the right-hand side.

Fix a positive-MMS agent $i<n$, and write $n-i+1=3m+r$, where $m\ge0$ and $r\in\{1,2,3\}$. By Lemma~\ref{lem:goods-shortfall}, $v_i(X_i)\ge x$ whenever $v_i(R_i)\ge2(n-i+1)/3$. Thus, if $v_i(X_i)<x$, then
\[
\begin{aligned}
v_i(R_i)-(n-i+1)+m
&<\frac{2(n-i+1)}3-(n-i+1)+m\\
&=-\frac r3.
\end{aligned}
\]
Since $H(t)\ge(-t)_+$, this event implies
\[
H(v_i(R_i)-(n-i+1)+m)
\ge(n-i+1)-m-v_i(R_i)
>\frac r3.
\]
By Markov's inequality and then~\eqref{eq:goods-probability-moment}, we therefore have
\[
\begin{aligned}
\Prb[v_i(X_i)<x]
&\le\Prb\left[
H(v_i(R_i)-(n-i+1)+m)>\frac r3
\right]\\
&\le\frac3r\cdot
\E[H(v_i(R_i)-(n-i+1)+m)]\\
&\le\frac3r\cdot
\left[
hz^m+\frac{11}{9}\cdot
\frac{z^{m+(n-i+4)/2}}{1-\sqrt z}
\right].
\end{aligned}
\]

For the final agent with positive MMS, $X_n=R_n$. Whenever $v_n(X_n)<x$, we have
\[
H(v_n(R_n)-1)\ge1-v_n(R_n)>1-x.
\]
Applying Markov's inequality and then~\eqref{eq:goods-probability-moment} with $\ell=1$ and $m=0$, we obtain
\[
\begin{aligned}
\Prb[v_n(X_n)<x]
&\le\Prb[H(v_n(R_n)-1)>1-x]\\
&\le\frac1{1-x}\cdot\E[H(v_n(R_n)-1)]\\
&\le\frac1{1-x}\cdot
\left[
h+\frac{11}{9}\cdot\frac{z^2}{1-\sqrt z}
\right].
\end{aligned}
\]

We next sum these probability bounds to bound $\Prb[\rho(X)<x]$. We first sum the terms arising from $H(m)=hz^m$. For a nonfinal agent with $n-i+1=3m+r$, this term is $(3/r)\cdot hz^m$. Thus, for $r=1,2,3$, the respective terms are $3hz^m$, $(3/2)\cdot hz^m$, and $hz^m$.

When $m=0$, the case $r=1$ corresponds to the final agent, whose bound uses $h/(1-x)$ instead of $3h$. The cases $r=2,3$ contribute $(3/2)\cdot h$ and $h$. For each $m\ge1$, we sum the three terms listed above. Each value of $n-i+1$ occurs at most once, and including values greater than $n$ only increases the sum because every term is nonnegative. Therefore, the total contribution from the $H(m)$ terms is at most
\[
\begin{aligned}
&\frac h{1-x}+\frac32\cdot h+h
+\sum_{m=1}^{\infty}\left(3+\frac32+1\right)\cdot hz^m\\
&\quad=h\left[\frac1{1-x}+\frac52+\frac{11}{2}\cdot\frac z{1-z}\right]\\
&\quad=h\left[\frac{11/2}{1-z}+\frac1{1-x}-3\right].
\end{aligned}
\]
We have summed the $H(m)$ terms. We now sum the remaining terms in the agents' probability bounds. For a nonfinal agent with $\ell=n-i+1=3m+r$, the remaining term is
\[
\frac3r\cdot\frac{11}{9(1-\sqrt z)}\cdot z^{m+(\ell+3)/2}
=
\frac{11z^2}{9(1-\sqrt z)}
\cdot z^{5m/2}\cdot\frac3r\cdot(\sqrt z)^{r-1}.
\]
For $r=1,2,3$, the last factor is respectively $3$, $(3/2)\cdot\sqrt z$, and $z$. When $m=0$, the case $r=1$ corresponds to the final agent and uses $1/(1-x)$ instead of $3$, just as in the preceding sum.

Including all possible values of $\ell$ therefore bounds the remaining terms by
\[
\begin{aligned}
&\frac{11z^2}{9(1-\sqrt z)}
\left[
\frac1{1-x}+\frac32\cdot\sqrt z+z
+\sum_{m=1}^{\infty}
\left(3+\frac32\cdot\sqrt z+z\right)\cdot z^{5m/2}
\right]\\
&\quad=
\frac{11z^2}{9(1-\sqrt z)}
\left[
\frac1{1-x}+\frac32\cdot\sqrt z+z
+\left(3+\frac32\cdot\sqrt z+z\right)
\cdot\frac{z^{5/2}}{1-z^{5/2}}
\right]\\
&\quad=
\frac{11z^2}{9(1-\sqrt z)}
\left[
\frac{3+(3/2)\cdot\sqrt z+z}{1-z^{5/2}}
+\frac1{1-x}-3
\right].
\end{aligned}
\]

Adding this bound to the preceding bound for the $H(m)$ terms gives an upper bound on the sum of the agents' probabilities of receiving value below $x$. By the union bound already established, this also bounds $\Prb[\rho(X)<x]$. Substituting $h=(117/8)\cdot x^2$ and $z=(9/4)\cdot x$, we obtain
\[
\begin{aligned}
\Prb[\rho(X)<x]\le x^2\Biggl\{
&\frac{117}{8}
\left[\frac{11/2}{1-z}+\frac1{1-x}-3\right]\\
&+\frac{99}{16(1-\sqrt z)}
\left[
\frac{3+(3/2)\cdot\sqrt z+z}{1-z^{5/2}}
+\frac1{1-x}-3
\right]
\Biggr\}.
\end{aligned}
\]
Here $99/16=(11/9)\cdot(9/4)^2$ comes from the factor $11z^2/9$ in the second bound.

Both inner brackets are positive and increase with $x$, and $1/(1-\sqrt z)$ also increases with $x$. Thus the coefficient of $x^2$ is largest at $x=1/15$, where $z=3/20$. Substitution at this endpoint gives a coefficient less than $85$. Consequently,
\[
\Prb[\rho(X)<x]<85\cdot x^2,
\]
which completes the proof.
\end{proof}

\subsection{An upper bound for arbitrary online algorithms}\label{subsec:goods-hardness}
In this section, we provide a hardness result for goods. The hidden-support instance from Theorem~\ref{thm:binary-hard-instance} also bounds the expected goods ratio from above. Its distribution is unchanged. For goods, too many hidden-support items allocated to early agents leave some later agent with value zero.

\begin{proof}[Proof of Theorem~\ref{thm:goods-hardness}]
Let $|M|=2n$ and choose $C$ uniformly among its $n$-element subsets. The first $n/2$ agents value every item at one and have MMS two; the remaining $n/2$ agents value precisely the items of $C$ at one and have MMS one. Fix the algorithm's random seed. The early bundles are independent of $C$. If one has at most one item, then $\rho(X)\le1/2$ for every $C$. Otherwise their union $U$ contains at least $n$ items; choose any $U'\subseteq U$ with $|U'|=n$.

If $|C\setminus U'|<n/2$, fewer than $n/2$ valuable items remain for the last $n/2$ agents, so one receives value zero. On every other outcome, $\rho(X)\le1$. The hypergeometric variable $|C\setminus U'|$ is symmetric about $n/2$, and its probability of equaling $n/2$ is $\binom n{n/2}^{2}/\binom{2n}n$. Therefore
\[
 \E_C[\rho(X)]
 \le\Prb[|C\setminus U'|\ge n/2]
 =\frac12+\frac12\cdot\frac{\binom n{n/2}^{2}}{\binom{2n}n}
 =\frac12+O\!\left(\frac1{\sqrt n}\right).
\]
Averaging over the random seed and then over $C$ yields one fixed instance with the same upper bound. The Stirling estimate is the same as in Theorem~\ref{thm:binary-hard-instance}.
\end{proof}

\end{document}